\documentclass[lettersize,journal]{IEEEtran}
\usepackage{amsmath,amssymb,amsfonts}
\usepackage{algorithmic}
\usepackage{algorithm}
\usepackage{array}
\usepackage[caption=false,font=normalsize,labelfont=sf,textfont=sf]{subfig}
\usepackage{textcomp}
\usepackage{stfloats}
\usepackage{url}
\usepackage{verbatim}
\usepackage{graphicx,psfrag,epsfig}
\usepackage{cite}
\usepackage{amssymb}
\usepackage{amsmath,lipsum}
\usepackage{cuted}
\usepackage{balance}
\usepackage{arydshln}
\usepackage{bm}
\usepackage{amsmath}
\usepackage{amsthm}
\usepackage{subeqnarray}
\usepackage{cases}
\usepackage{xcolor}
\usepackage{soul}
\sethlcolor{green}
\usepackage{xcolor}

\graphicspath{{./}}
\newtheorem{Theorem}{Theorem}
\newtheorem{Lemma}{Lemma}

\newtheorem*{remark}{Remark}
\def\linspread{0.95}
\def\linspreadalgr{0.75}
\IEEEaftertitletext{\vspace{-10pt}}

\begin{document}

\title{Antenna Position and Beamforming Optimization for Movable Antenna Enabled ISAC: Optimal Solutions and Efficient Algorithms}

\author{Lebin Chen, \IEEEmembership{Student Member,~IEEE,} Ming-Min Zhao,
	\IEEEmembership{Senior Member,~IEEE}, Min-Jian Zhao, \IEEEmembership{Member,~IEEE}, and Rui Zhang, \IEEEmembership{Fellow,~IEEE}
	\thanks{L. Chen, M. M. Zhao, and M. J. Zhao are with the College of Information Science and Electronic
		Engineering, Zhejiang University, Hangzhou 310027, China, and also with the Zhejiang Provincial Key Laboratory of Multi‐Modal Communication Networks and Intelligent Information Processing, Hangzhou 310027, China (e-mail:
		12431101@zju.edu.cn; zmmblack@zju.edu.cn; mjzhao@zju.edu.cn). (\emph{Corresponding author: Ming-Min Zhao}.)}
	\thanks{
		R. Zhang is with School of Science and Engineering, Shenzhen Research Institute of Big Data, The Chinese University of Hong Kong, Shenzhen, Guangdong 518172, China (e-mail: rzhang@cuhk.edu.cn). He is also with the Department of Electrical and Computer Engineering, National University of Singapore, Singapore 117583 (e-mail: elezhang@nus.edu.sg).
		The work of Rui Zhang is supported in part by The Guangdong Provincial Key Laboratory of Big Data Computing, the National Natural Science Foundation of China (No. 62331022), the 2022 Stable Research Program of Higher Education of China (No. 20220817144726001), and the Guangdong Major Project of Basic and Applied Basic Research (No. 2023B0303000001).}
}



\maketitle

\begin{abstract}
In this paper, we propose an integrated sensing and communication (ISAC) system enabled by movable antennas (MAs), which can dynamically adjust antenna positions to enhance both sensing and communication performance for future wireless networks.	
To characterize the benefits of MA-enabled ISAC systems, we first derive the Cramér-Rao bound (CRB) for angle estimation error, which is then minimized for optimizing the antenna position vector (APV) and beamforming design, subject to a pre-defined signal-to-noise ratio (SNR) constraint to ensure the communication performance.
In particular, for the case with receive MAs only, we provide a closed-form optimal antenna position solution, and show that employing MAs over conventional fixed-position antennas (FPAs) can achieve a sensing performance gain upper-bounded by 4.77 dB.
On the other hand, for the case with transmit MAs only, we develop a boundary traversal breadth-first search (BT-BFS) algorithm to obtain the global optimal solution in the line-of-sight (LoS) channel scenario, along with a lower-complexity boundary traversal depth-first search (BT-DFS) algorithm to find a local optimal solution efficiently. While in the scenario with non-LoS (NLoS) channels, a majorization-minimization (MM) based Rosen's gradient projection (RGP) algorithm with an efficient initialization method is proposed to obtain stationary solutions for the considered problem, which can be extended to the general case with both transmit and receive MAs. 
Extensive numerical results are presented to verify the effectiveness of the proposed algorithms, and demonstrate the superiority of the considered MA-enabled ISAC system over conventional ISAC systems with FPAs in terms of sensing and communication performance trade-off.
\end{abstract}

\vspace{-0cm}
\section{Introduction}
\IEEEPARstart{I}{ntegrated} sensing and communication (ISAC) has been identified as a key enabling technology for the future/six-generation (6G) wireless networks, with the ability to offer unified communication and sensing functionalities in a single system\cite{6gWCsystems}. This integration is particularly important for emerging applications such as autonomous driving, smart cities, and immersive virtual reality, all of which require not only high-speed data transmission but also precise sensing capabilities\cite{Avisionof6g}. 
ISAC enables simultaneous sensing and communication by leveraging shared hardware and spectrum, which reduces deployment costs and energy consumption. Moreover, ISAC offers potential performance improvements through joint optimization of communication and sensing tasks, allowing for dynamic trade-offs depending on real-time requirements\cite{intro16}. As a result, ISAC is expected to play a critical role in the deployment of future wireless systems, where diverse and dynamic demands will need to be met seamlessly\cite{ISACsurvey}.

In many communication and sensing systems, fixed-position antenna (FPA) arrays, such as uniform linear/planar arrays (ULAs/UPAs) and sparse arrays\cite{sparse}, are usually employed. While FPAs are widely used in existing applications\cite{intro7,intro9,ap1},
their limitations become evident in the context of ISAC, where flexibility and adaptability are crucial. FPAs consist of antennas fixed in place with pre-defined geometries, which restrict their ability to adapt to varying sensing and communication environments \cite{ma2024movableantennaenhancedwireless,fpa1,fpa2}.
One major drawback of FPAs is their inability to optimize the antenna configuration in real-time. In wireless sensing tasks, the system performance such as radar resolution, heavily depends on the geometry of the antenna array. 
From a communication perspective, FPAs also face significant limitations, especially in multiple-input multiple-output (MIMO) systems, since the channel capacity in MIMO systems is closely tied to the spatial configuration of antennas, and FPAs are unable to leverage the full potential of spatial diversity due to their static positions\cite{10243545,MAsurvey}. This rigidity constrains their ability to dynamically reshape the MIMO channel matrix to improve the communication performance. In rich multipath environments, where spatial diversity is crucial for improving receiver signal-to-noise ratio (SNR) and channel capacity, FPAs are unable to adapt to substantial spatial/time variations of wireless channels and thus may result in poor  performance. Additionally, attempts to improve the performance by increasing the number of antennas lead to higher hardware costs and energy consumption, making large-scale FPA deployment practically challenging in many real-world scenarios\cite{10243545}.

Fortunately, movable antenna (MA) system, also known as fluid antenna system (FAS) \cite{zhu2024historicalreviewfluidantenna}, offers a promising alternative solution to overcome the limitations of FPAs in ISAC applications. Unlike FPAs, MAs can dynamically adjust their positions, allowing for real-time optimization of both communication and sensing tasks. This new degree of freedom is envisioned to significantly enhance the system performance, particularly in complex and dynamic environments.
In wireless sensing, MAs provide substantial improvements in tasks such as angle of arrival (AoA) estimation and target detection \cite{ma2024movableantennaenhancedwireless}. By allowing antenna elements to move and reconfigure, MA systems can increase the effective aperture of the array by spreading the antennas over a larger area, resulting in higher angular resolution without increasing the number of antennas, thus maximizing antenna cost-efficiency\cite{ma2024movableantennaenhancedwireless}.
For communication systems, MA-enabled MIMO systems offer significant advantages in terms of channel capacity and spatial multiplexing gain\cite{10243545}. By optimizing the positions of both transmit and receive antennas, MA systems can reconfigure the MIMO channel matrix to enhance channel quality and improve data throughput. This adaptability allows MA systems to maximize the use of spatial diversity in multipath environments, leading to higher capacity compared to traditional FPA systems\cite{MA2}. 
Additionally, MA systems are particularly beneficial in multi-user communication environments, where precise control over antenna positions can help mitigate interference and improve performance for all users\cite{MA5}.
Furthermore, the dynamic nature of MA systems enables them to switch seamlessly between communication and sensing modes, making them ideal for ISAC applications. This flexibility is especially crucial in dynamic scenarios, such as vehicle-to-everything (V2X) communication networks, where both high-performance sensing and reliable communication are necessary\cite{MAsurvey}. 
{{To achieve this flexibility, MAs adaptively respond to changing communication and sensing conditions, enabling superior performance even with fewer antennas. By leveraging dynamic reconfigurability and sub-wavelength positioning, MAs provide highly directive beamforming, precise angular resolution, and effective interference suppression 
\cite{R1,R2}. 
Recent studies confirm that MA-based arrays consistently outperform fixed-position arrays through multiple mechanisms, including reducing the Cramér–Rao bound (CRB) in angle estimation \cite{FASforISAC1}, enabling joint optimization of antenna positions and beamforming to enhance communication and sensing performance \cite{R3}, and providing flexible beampattern shaping capabilities in bistatic radar ISAC scenarios \cite{R4}.
Beyond linear antenna array optimization, recent works have also demonstrated that exploiting two-dimensional/three-dimensional (2D/3D) array geometry \cite{R1,R2}, and 3D antenna rotation adjustments can further enhance communication/sensing performance in six-dimensional MA (6DMA) systems  \cite{6DMA1,6DMA2,R5,R6}.}}

MA-enabled ISAC systems have gained considerable attention for future 6G networks, as they provide enhanced flexibility and performance. MA-enabled wireless sensing was first studied in \cite{ma2024movableantennaenhancedwireless}, by focusing on movable receive antennas. While this study provided valuable insights into the benefits of receive MAs for sensing performance enhancement, it did not explore the potential advantages of movable transmit antennas. Subsequently, \cite{FASforISAC1} and \cite{FASforISAC2} extended the MA-ISAC framework by incorporating movable transmit antennas and employing alternating optimization based algorithm to optimize beamforming and antenna positions, while \cite{FASforISAC3} leveraged the deep reinforcement learning (DRL) technique.
However, these approaches have their limitations, primarily due to their reliance on local optimization methods and the absence of optimal strategies for the joint optimization of transmit and receive MAs.

Motived by the above discussions, we focus on developing optimal solutions and efficient algorithms for MA-enabled ISAC systems in this work, where the CRB is used as the objective function to optimize the antenna position vector (APV) and beamforming vectors, while a pre-defined SNR level is imposed as a constraint to guarantee communication performance.
The main contributions of this work are summarized as follows:
\begin{itemize}
	\item First, for the case with receive MAs only, we derive an optimal closed-form solution for the antenna positions, despite the non-convexity of the objective function and variable-coupling issue. 
	Besides, we show that the sensing performance gain 
	obtained via receive MAs  is upper-bounded by 4.77 dB.
	\item Next, for the case with transmit MAs only, we propose a boundary traversal breadth-first search (BT-BFS) algorithm to obtain the global optimal solution of the considered problem in the line-of-sight (LoS) channel scenario.
	Besides, to reduce the complexity of the BT-BFS algorithm, a boundary traversal depth-first search (BT-DFS) algorithm is further proposed to obtain a local optimal solution with only linear computational complexity.  
	Then, for the non-line-of-sight (NLoS) channel scenario, we propose a majorization-minimization (MM) based Rosen's gradient projection (RGP) algorithm to obtain a stationary solution for the considered problem, together with a meticulously designed method to attain a high-quality initial solution.
	Furthermore, we theoretically show that utilizing movable transmit antenna array is able to enlarge the SNR-CRB region, and the corresponding performance gain is determined by the directional channel response, which is defined as the magnitude of the inner product between the channel vector and the steering vector.
	\item Finally, we show that the general case with MAs for both receiving and transmitting can be simply addressed by combining the above techniques. 
	Numerical results are presented to show the effectiveness of the proposed algorithms and validate the significant ISAC performance gain achieved by using MA array over traditional FPAs. 
\end{itemize}

The remainder of this paper is organized as follows. Section \uppercase\expandafter{\romannumeral2} introduces the system model and problem formulation. Sections \uppercase\expandafter{\romannumeral3} and \uppercase\expandafter{\romannumeral4} provide the joint antenna position and beamforming designs for the cases with receive/transmit MAs, respectively. Section \uppercase\expandafter{\romannumeral4}
provides numerical results. Finally, conclusions are drawn in Section \uppercase\expandafter{\romannumeral5}.

\begin{figure}[t]
	\centering
	\includegraphics[width=3.0in]{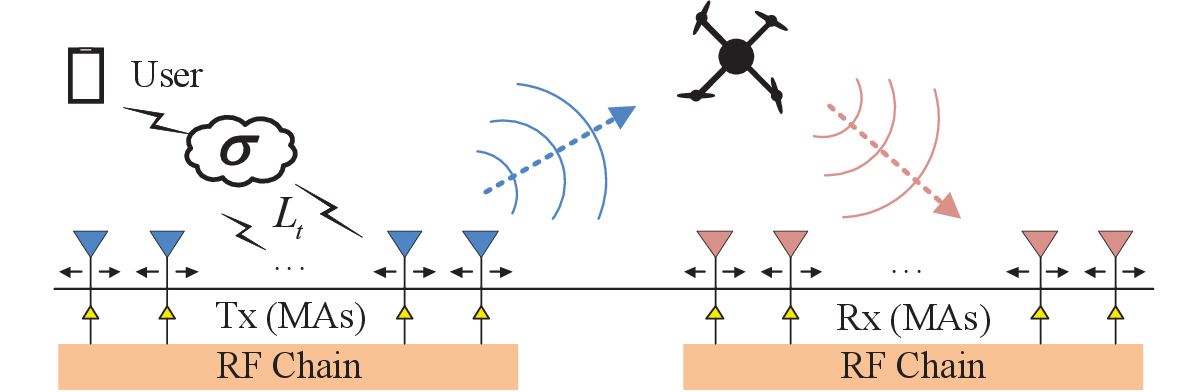}\vspace{-0.2cm}
	\caption{An MA-enabled ISAC system.}
	\label{model}
	\vspace{-0cm}
\end{figure}

\textit{Notations:}
Scalars, vectors, and matrices are denoted by lower/upper case, bold-face lower-case, and bold-face uppercase letters, respectively.
$(\cdot)^T$, $(\cdot)^*$ and $(\cdot)^H$ denote the transpose, conjugate and conjugate transpose operators, respectively. We use $\operatorname{Re}\{x\}$ and $\angle x $ to denote the real part and the phase of a complex number $x$. $\mathbf{I}$, $\mathbf{J}$, $\mathbf{0}$ and $\mathbf{1}$  are used to represent identity matrix, all-one matrix, all-zero vector and all-one vector with proper dimensions, respectively. 
$\begin{Vmatrix}\cdot\end{Vmatrix}$ denotes the Euclidean norm of a complex vector/matrix, and $|\cdot|$ denotes the absolute value of a complex number. 
$\mathbf{y} \backslash y_i$ denotes the set/vector obtained by removing the element $y_i$ from $\mathbf{y}$, where $\mathbf{y}$ is a set or vector and $y_i$ is one of its elements. 
The set of integers is denoted by $\mathbb{Z}$, and the sets of $P \times Q$ dimensional complex, real and positive real matrices are denoted by $\mathbb{C}^{P\times Q}$, $\mathbb{R}^{P\times Q}$ and $\mathbb{R}^{P\times Q}_{++}$, respectively.

\section{System Model And Problem Formulation}
\subsection{System Setting}
In this work, we consider a MIMO ISAC base station (BS) with $N_t$ transmit MAs and $N_r$ receive MAs, which is serving a downlink user with a single antenna while detecting a single point target, as depicted in Fig. \ref{model}. 
In order to avoid information loss of the sensed target, the number of receive antennas is assumed to be greater than that of transmit antennas, i.e., $N_r>N_t$ \cite{9652071}. Additionally, the transmit and receive MAs can be dynamically repositioned in real-time through flexible cables linked to the RF chains \cite{MAsurvey}. Furthermore, we assume narrow-band quasi-static channels, such that the time overhead for adjusting MA positions is tolerable compared to the much longer channel coherence time of both the user and target\cite{MAsurvey}. 
\footnote{We assume that the user/target is static for a long period or moves at a much lower speed as compared to that of the antenna movement.}   
As shown in Fig. \ref{model}, we consider linear MA arrays of sizes $N_t$ and $N_r$ at the transmitter and receiver sides, respectively. Denote $\mathbf{x}=[x_1,x_2,...,x_{N_t}]^T\in\mathbb{R}^{N_t}$ and $\mathbf{y}=[y_1,y_2,...,y_{N_r}]^T\in\mathbb{R}^{N_r}$ as the APVs of the transmit and receive MA arrays, respectively. The steering vectors of the transmit and receive MA arrays can be respectively written as \vspace{-0.2cm}
\begin{equation}
	\label{s1}
	\small
	\mathbf{a}(\mathbf{x},\theta)=\begin{bmatrix}e^{-j\frac{2\pi}{\lambda}x_1\sin\theta},e^{-j\frac{2\pi}{\lambda}x_2\sin\theta},...,e^{-j\frac{2\pi}{\lambda}x_{N_t}\sin\theta}\end{bmatrix}^T,
\end{equation}
 \begin{equation}
 	\label{s2}\small
 	\mathbf{b}(\mathbf{y},\theta)=\begin{bmatrix}e^{-j\frac{2\pi}{\lambda}y_1\sin\theta},e^{-j\frac{2\pi}{\lambda}y_2\sin\theta},...,e^{-j\frac{2\pi}{\lambda}y_{N_r}\sin\theta}\end{bmatrix}^T,
 \end{equation}
where $\lambda$ denotes the carrier wavelength and $\theta$ is the azimuth angle of the point target relative to the BS.
{{Due to the
	co-located radar setting, the direction of arrival (DoA) and
	the direction of departure (DoD) are the same for the same
	target.}}

In our considered system, the channel vector is determined by the signal propagation environment and the positions of the transmit MAs. We consider the field-response based channel model \cite{10243545}, where the number of transmit paths is denoted as $L_t$. The azimuth angle of departure (AoD) of the $k$th ($k=1,2,...,L_t$) transmit path is denoted by $\theta_t^k\in[0,\pi]$. Thus, the field response vector of the $i$th  transmit MA can be defined as $\mathbf{g}(x_i)=[e^{-j\frac{2\pi}{\lambda}x_i\sin\theta_t^1},e^{-j\frac{2\pi}{\lambda}x_i\sin\theta_t^2},...,e^{-j\frac{2\pi}{\lambda}x_i\sin\theta_t^{L_t}}]^T$.
By stacking $\mathbf{g}(x_i)$ of all $N_t$ transmit MAs, the field response matrix of the transmit MA array is given by $\mathbf{G}(\mathbf{x})=[\mathbf{g}(x_1),\mathbf{g}(x_2),...,\mathbf{g}(x_{N_t})]\in\mathbb{C}^{L_t\times{N_t}}$.
{Furthermore, we define the channel gain vector as $\bm{\sigma}\in\mathbb{C}^{L_t\times1}$, where  $\sigma_k$ is the complex gain of the $k$-th propagation path, which accounts for both path loss and phase shift.}
As a result, the channel vector from the transmitter to the downlink user is given by\vspace{-0.1cm}
\begin{equation}
	\label{q18}\small
	\mathbf{h}^H(\mathbf{x})=\bm{\sigma}^H\mathbf{G}(\mathbf{x}),
	\vspace{-0.1cm}
\end{equation}
which is assumed to be known to the BS \cite{10243545,9652071}.
Let $\mathbf{X}\in\mathbb{C}^{N_t\times{L}}$ be a narrowband ISAC signal matrix, with $L>N_t$ being the length of the radar pulse/communication frame. 
Furthermore, let $\mathbf{X}=\mathbf{w}\mathbf{s}^H$,
where $\mathbf{w}$ is the dual-functional beamforming vector to be designed, and $\mathbf{s}\in\mathbb{C}^{L\times1}$ contains the unit-power data stream intended for the communication user. The data stream $\mathbf{s}$ is assumed to have independent entries and satisfy $\frac{1}{L}\mathbf{s}^H\mathbf{s}=1$.
By transmitting $\mathbf{X}$ to the communication user, the received signal vector is \vspace{-0.1cm}
\begin{equation}
	\mathbf{y}_C^H=\mathbf{h}^H\mathbf{X}+\mathbf{z}_C^H,
	\vspace{-0.1cm}
\end{equation}
where $\mathbf{z}_C\in\mathbb{C}^{{L}\times1}$ is an additive white Gaussian noise (AWGN) vector with zero mean and covariance matrix $\sigma_C^2 \mathbf{I}$. 
By transmitting $\mathbf{X}$ to sense the target, the reflected echo signal at the receiver of the ISAC BS is given by \vspace{-0.1cm}
\begin{equation}
	\mathbf{Y}_R=\mathbf{B}\mathbf{X}+\mathbf{Z}_R,
	\vspace{-0.1cm}
\end{equation}
where  $\mathbf{Z}_R\in\mathbb{C}^{N_r\times{L}}$ denotes an AWGN matrix with zero mean and the variance of each its entry is $\sigma_R^2$, $\mathbf{B}\in\mathbb{C}^{N_r\times{N_t}}$ represents the target response matrix. Since the target is modeled as an unstructured point that is far away from the BS, the target response matrix $\mathbf{B}$ can be written as $\mathbf{B}=\alpha\mathbf{b}(\mathbf{y},\theta)\mathbf{a}^H(\mathbf{x},\theta)
\triangleq\alpha\mathbf{A}(\mathbf{x},\mathbf{y},\theta)$,
where $\alpha\in\mathbb{C}$ represents the reflection coefficient, which contains both the round-trip path-loss and the radar cross-section (RCS) of the target.

\vspace{-0.2cm}
\subsection{Problem Formulation}
In this paper, we aim to optimize the sensing performance under the communication quality-of-service (QoS) constraint by jointly optimizing the transmit and/or receive MA positions ${\mathbf{x}}$, ${\mathbf{y}}$, and the transmit beamforming vector ${\mathbf{w}}$.
In particular, we adopt the CRB for evaluating the target estimation performance, which is a lower bound on the variance of any unbiased estimators, and employ the user SNR to measure the communication performance.

For the considered point target case, the \text{CRB} for $\theta$ was derived in \cite{Kay1993FundamentalsOS}, \cite{1703855}, and is given by (\ref{q2}) at the bottom of the next page, where
\begin{figure*}[hb] 
	\centering 
	\vspace*{-12pt}
	\hrulefill 
	\vspace*{-5pt} 
	\begin{equation}
		\label{q2}\small
		\begin{aligned}
			\text{CRB}(\theta)=\frac{\sigma_R^2\text{tr}\Big(\mathbf{A}^H(\theta)\mathbf{A}(\theta)\mathbf{R}_X\Big)}
			{
				2\begin{vmatrix}
					\alpha
				\end{vmatrix}^2L\Big(\text{tr}\big(\dot{\mathbf{A}}^H(\theta)\dot{\mathbf{A}}(\theta)\mathbf{R}_X\big)
				\text{tr}\big(\mathbf{A}^H(\theta)\mathbf{A}(\theta)\mathbf{R}_X\big)
				-\begin{vmatrix}
					\text{tr}\big(\dot{\mathbf{A}}^H(\theta){\mathbf{A}}(\theta)\mathbf{R}_X\big)
				\end{vmatrix}^2\Big)
			}.
		\end{aligned}
	\end{equation}
\end{figure*}

\vspace{-0.3cm}
\begin{equation}
	\label{q3}\small
	\mathbf{R}_X=\frac{1}{L}\mathbf{X}\mathbf{X}^H=
	\frac{1}{L}\mathbf{w}\mathbf{s}^H\mathbf{s}\mathbf{w}^H=
	\mathbf{w}\mathbf{w}^H,
	\vspace{-0.2cm}
\end{equation}
is the sample covariance matrix of $\mathbf{X}$, and $\dot{\mathbf{A}}\triangleq\frac{\partial\mathbf{A}(\theta)}{\partial\theta}$.
Furthermore, the receive SNR at the user is given as
$	\gamma=\frac{|\mathbf{h}^H\mathbf{w}|^2}{\sigma_C^2}$.
Thus, the corresponding joint optimization problem 
can be formulated as \vspace{-0.1cm}
\begin{subequations}
	\label{p1}\small
	\begin{align}
		\min\limits_{\mathbf{x},\mathbf{y},\mathbf{w}}&\;\;\text{CRB}(\theta)\\
		\text{s.t.}&\;\;\gamma\ge\Gamma,
		\begin{vmatrix}
			\mathbf{w}
		\end{vmatrix}^2\le{P_T},\\
		&\;\;D_x\ge\begin{vmatrix}
			x_m-x_n
		\end{vmatrix}\ge{d},1\le{m},{n}\le{N_t},m\neq{n}, \label{q4}\\
		&\;\;D_y\ge\begin{vmatrix}
			y_m-y_n
		\end{vmatrix}\ge{d},1\le{m},{n}\le{N_r},m\neq{n}, \label{q5}
	\end{align}
\end{subequations}
where $\Gamma$ is the required SNR level for the user, $P_T$ is the transmit power budget, $d$ is the minimum distance (usually we set $d=\frac{\lambda}{2}$) between any two MAs to avoid the coupling effect, and $D_x$ ($D_y$) is the aperture of the overall transmit (receive) antenna array. Due to the equivalence between the antennas, we can assume $x_1<x_2<...<x_{N_t}$ and $y_1<y_2<...<y_{N_r}$ without loss of optimality. Thus, constraints (\ref{q4}) and (\ref{q5}) can be equivalently transformed into linear inequality constraints $\mathbf{U}\mathbf{x}\preceq\mathbf{l}_u$ and $\mathbf{V}\mathbf{y}\preceq\mathbf{l}_v$, respectively,
where \vspace{-0.1cm}
\begin{equation}
	\label{q16}\small
	\begin{aligned}
	\mathbf{U}&=\left[\begin{array} {ccccccc}
		1 & -1 & 0 & 0 & \cdots & 0 & 0\\
		0 & 1 & -1 & 0 & \cdots & 0 & 0\\
		\vdots & \vdots & \vdots & \vdots & \ddots & \vdots & \vdots\\
		0 & 0 & 0 & 0 & \cdots & 1 & -1\\
		-1 & 0 & 0 & 0 & \cdots & 0 &1
	\end{array}\right]_{N_t\times{N_t}},\\
			\mathbf{l}_u&=\left[
	-d,  -d ,\cdots , -d , D_x
	\right]^T_{N_t\times1},
\end{aligned}
\end{equation}
and $\{\mathbf{V},\mathbf{l}_v\}$ are similarly defined.
 
Based on the properties of the steering vectors in (\ref{s1}), (\ref{s2}) and their derivatives, it can be easily verified that \vspace{-0.1cm}
\begin{equation}
	\label{q6}\small
	\dot{\mathbf{a}}^H(\mathbf{x},\theta)\mathbf{a}(\mathbf{x},\theta)=j\frac{2\pi}{\lambda}\cos\theta
	\sum_{i=1}^{N_t}x_i,\forall\theta,
\end{equation}
\vspace{-0.6em}
\begin{equation}
	\label{q7}\small
	\dot{\mathbf{b}}^H(\mathbf{y},\theta)\mathbf{b}(\mathbf{y},\theta)=
	j\frac{2\pi}{\lambda}\cos\theta
	\sum_{i=1}^{N_r}y_i,\forall\theta.
	\vspace{-0.1cm}
\end{equation}
Leveraging (\ref{q3}), (\ref{q6}) and (\ref{q7}) yields (\ref{q8})-(\ref{q9}), i.e.,\vspace{-0.1cm}
\begin{equation}
	\label{q8}\small
	\text{tr}\big({{\mathbf{A}^H}\mathbf{A}{\mathbf{R}_X}}\big) =
	 {N_r}{\begin{vmatrix}{{\mathbf{a}^H}{\mathbf{w}}}\end{vmatrix}^2},
\end{equation}
\vspace{-0.5cm}
\begin{equation}\small
	\text{tr}\big({{\dot{\mathbf{A}}^H}\mathbf{A}{\mathbf{R}_X}}\big) =
	j\frac{2\pi}{\lambda}\cos\theta
	{\begin{vmatrix}{{\mathbf{a}^H}{\mathbf{w}}}\end{vmatrix}^2}
	\sum_{i=1}^{N_r}y_i
	+N_r\mathbf{a}^H\mathbf{w}\mathbf{w}^H\dot{\mathbf{a}},
	\vspace{-0.1cm}
\end{equation}
\begin{figure*}[hb] 
	\centering 
	\vspace*{-28pt}
	\begin{equation}\small
		\label{q9}
		\text{tr}\big({{\dot{\mathbf{A}}^H}\dot{\mathbf{A}}{\mathbf{R}_X}}\big) =
		(\frac{2\pi}{\lambda}\cos\theta)^2
		{\begin{vmatrix}{{\mathbf{a}^H}{\mathbf{w}}}\end{vmatrix}^2}
		\sum_{i=1}^{N_r}y_i^2
		+j\frac{2\pi}{\lambda}\cos\theta
		\sum_{i=1}^{N_r}y_i
		(\dot{\mathbf{a}}^H\mathbf{w}\mathbf{w}^H{\mathbf{a}}-\mathbf{a}^H\mathbf{w}\mathbf{w}^H\dot{\mathbf{a}})
		+{N_r}{\begin{vmatrix}{{\dot{\mathbf{a}}^H}{\mathbf{w}}}\end{vmatrix}^2}.
	\end{equation}
	\vspace*{-16pt} 
\end{figure*}where $\mathbf{A}\triangleq\mathbf{A}(\mathbf{x},\mathbf{y},\theta)$ and $\dot{\mathbf{A}}\triangleq\dot{\mathbf{A}}(\mathbf{x},\mathbf{y},\theta)$, and (\ref{q9}) is shown at the bottom of the next page. Substituting (\ref{q8})-(\ref{q9}) into (\ref{q2}), the objective function of problem (\ref{p1}), i.e., $\text{CRB}(\theta)$, can be simplified to\vspace{-0.2cm}
\begin{equation}
	\label{CRB_expression}\small
	\begin{aligned}
		\text{CRB}(\theta)=
		\frac{\sigma_R^2/
			(2\begin{vmatrix}
				\alpha
			\end{vmatrix}^2L)}
		{
			(\frac{2\pi}{\lambda}\cos\theta)^2
			{\begin{vmatrix}{{\mathbf{a}^H}{\mathbf{w}}}\end{vmatrix}^2}
			\big(\sum\limits_{i=1}^{N_r}y_i^2
			-\frac{1}{N_r}(\sum\limits_{i=1}^{N_r}y_i)^2\big)
		}.
	\end{aligned}\vspace{-0.1cm}
\end{equation}
Thus, problem (\ref{p1}) can be equivalently recast as\vspace{-0.2cm}
\begin{subequations}
	\label{p2}\small
\begin{align}
(\text{P})	\quad \quad	\max\limits_{\mathbf{x},\mathbf{y},\mathbf{w}}\;\;
&{\begin{vmatrix}{{\mathbf{a}^H}{\mathbf{w}}}\end{vmatrix}^2}
\left(\sum\limits_{i=1}^{N_r}y_i^2
-\frac{1}{N_r}\left(\sum\limits_{i=1}^{N_r}y_i\right)^2\right)\label{p2a}\\
\text{s.t.}\;\;&\begin{vmatrix}{{\mathbf{h}^H}{\mathbf{w}}}\end{vmatrix}^2
\ge\Gamma\sigma_C^2,
\|
\mathbf{w}
\|^2\le{P_T},\\
&\mathbf{U}\mathbf{x}\preceq\mathbf{l}_u,\mathbf{V}\mathbf{y}\preceq\mathbf{l}_v.
	\end{align}
\end{subequations}
In this work, we aim to obtain the optimal solution of problem (P), i.e., derive the optimal beamforming vector and antenna positions for the considered MA-assisted ISAC system. However, problem (P) is a highly non-convex optimization problem and very difficult to solve due to the coupling between the beamforming and antenna
position variables, as well as the fact that the position variables appear in the exponential terms of each channel coefficient. 
Generally, there is no efficient method for solving the non-convex problem (P) optimally.
To address the abovementioned difficulties, we first examine two special cases of problem (P) in the next two sections respectively, i.e., 1) only the receive antennas are movable, and 2) only the transmit antennas are movable. 
Then, for the general case where both receive and transmit antennas are movable, we show that the corresponding optimization problem can be efficiently solved by combining the techniques proposed in the above two special cases. 

 \vspace{-0.2cm}
\section{Receive MA Case}

In this section, we investigate the case with receive MAs only, and show that variable decoupling can be achieved by substituting the closed-form solution of the optimal beamforming vector (as a function of the APVs) into the objective function of problem (P), i.e., (\ref{p2a}). Then, we exploit the translational invariance and symmetry properties of (\ref{p2a}) over the receive MAs' APV $\mathbf{y}$ and provide the optimal closed-form solution to problem (P).
Moreover, we characterize the sensing performance gain provided by receive MAs, which is shown to be capped at 4.77 dB.
\vspace{-0.3cm}
\subsection{Optimal Solution}
In this subsection, a closed-form optimal solution to problem (P) is derived, where we assume that the transmit antennas are in the form of a uniform linear array (ULA) with half wavelength spacing. Under this assumption, problem (P) can be reduced to\vspace{-0.2cm} 
\begin{equation}
	\label{p3}\small
	\begin{aligned}
	(\text{AP})	\quad \quad	\max\limits_{\mathbf{y},\mathbf{w}}\;\;
		&{\begin{vmatrix}{{\mathbf{a}^H}{\mathbf{w}}}\end{vmatrix}^2}
		\underbrace{\left(\sum\limits_{i=1}^{N_r}y_i^2
		-\frac{1}{N_r}\left(\sum\limits_{i=1}^{N_r}y_i\right)^2\right)}_{\triangleq{f(\mathbf{y})}}\\
		\text{s.t.}\;\;&\begin{vmatrix}{{\mathbf{h}^H}{\mathbf{w}}}\end{vmatrix}^2\ge\Gamma\sigma_C^2,\\
		&\|
			\mathbf{w}
		\|^2\le{P_T},
		\mathbf{V}\mathbf{y}\preceq\mathbf{l}_v.
	\end{aligned} \vspace{-0.1cm}
\end{equation}
As can be seen, addressing problem (AP) is challenging even in this simplified scenario due to its highly non-convexity and the coupling between the variables $\mathbf{y}$ and $\mathbf{w}$. However, we also observe that $\mathbf{y}$ and $\mathbf{w}$ are decoupled in the constraints of problem (AP), and the objective function can be separated into two terms, i.e., ${|{{\mathbf{a}^H}{\mathbf{w}}}|}$ and $f(\mathbf{y})$. 
This is intuitive, as the sensing antennas (i.e., the receive MAs) do not affect the information transmission and thus do not relate to the transmit beamforming. 

To proceed, we obtain the optimal solution of $\mathbf{w}$ as \cite{9652071}\vspace{-0.1cm}
\begin{subequations}
	\label{q11}\small
	\begin{numcases}{\mathbf{w}=}
		\sqrt{P_T} {\mathbf{a}}/{\|\mathbf{a}\|}, 
		& \text{if $P_T\begin{vmatrix}{{\mathbf{h}^H}{\mathbf{a}}}\end{vmatrix}^2>N_t\Gamma\sigma_C^2$,}
		\label{q11a} \\
		c_1\mathbf{u}_1+c_2\mathbf{a}_u,
		& \text{otherwise,} \label{q11b}
\end{numcases}
\end{subequations}
where
$		\mathbf{u}_1=\frac{\mathbf{h}}{\|\mathbf{h}\|} $,
		$\mathbf{a}_u=\frac{\mathbf{a}-(\mathbf{u}_1^H\mathbf{a})\mathbf{u}_1}
		{\|\mathbf{a}-(\mathbf{u}_1^H\mathbf{a})\mathbf{u}_1\|}$,
	$	c_1=\sqrt{\frac{\Gamma\sigma_C^2}{\|\mathbf{h}\|^2}}
		\frac{\mathbf{u}_1^H\mathbf{a}}{\|\mathbf{u}_1^H\mathbf{a}\|}$ and 
		$c_2=\sqrt{P_T-\frac{\Gamma\sigma_C^2}{\|\mathbf{h}\|^2}}
		\frac{\mathbf{a}_u^H\mathbf{a}}{\|\mathbf{a}_u^H\mathbf{a}\|}.$
Consequently, problem (AP) can be equivalently transformed into the following problem without loss of optimality:
	\begin{flalign}
		\label{p4}\small
	\quad\quad\quad\quad\quad\;\; (\text{AP-m})	\quad \quad  &\max\limits_{\mathbf{y}}\;\;
		f(\mathbf{y}) \notag &&
		\\
		&\;\;\;\text{s.t.}\;\;\mathbf{V}\mathbf{y}\preceq\mathbf{l}_v. &&
	\end{flalign}
Although problem (AP-m) is still non-convex due to the non-convex $f(\mathbf{y})$, we will prove that it admits a closed-form optimal solution.\footnote{Problem (AP-m) is similar to the one considered in \cite{ma2024movableantennaenhancedwireless}, where mathematical induction is used to derive the corresponding optimal solution. However, different from  \cite{ma2024movableantennaenhancedwireless}, this paper presents a completely different proof based on the translational invariance and symmetry properties of $f(\mathbf{y})$.} 
Before providing the formal proof, we first present the following lemmas, which state the properties of \( f(\mathbf{y}) \) and the conditions that the optimal solution must satisfy.
\begin{Lemma}
	For any $a\in\mathbb{R}$, $f(\mathbf{y})$ must satisfy the following translational invariance property:
	$f(\mathbf{y})=f(\mathbf{y}+a)$.
\end{Lemma}

\begin{proof}
Please refer to Appendix A. 
\end{proof}

\begin{Lemma}
	For any $D_y\in\mathbb{R}$, $f(\mathbf{y})$ must satisfy the following symmetry property:
	$f(\mathbf{y})=f(D_y-\mathbf{y})$. 
\end{Lemma}

\begin{proof}
	This lemma can be similarly proved as Lemma 1, and the detailed proof is omitted here due to space limitation.
\end{proof}

\begin{Lemma}
The optimal solution of problem (AP-m) satisfies $y_i=y_{i-1}+d,\;2\le{i}\le\left\lfloor{\frac{N_r}{2}}\right\rfloor$.
\end{Lemma}

\begin{proof}
	Please see Appendix B. 
\end{proof}

Based on the above lemmas, we have the following theorem that provides the optimal antenna positions in the receive-MA-only case.

\begin{Theorem}
The optimal solution of problem (AP-m) is given by \vspace{-0.2cm}
\begin{equation}
	\label{q10}\small
	\mathbf{y}_\text{opt}=
	\left\{
	\begin{aligned}
		[
		&0,d,\cdots,(\frac{N_r}{2}-1)d,D_y-(\frac{N_r}{2}-1)d,\\
		&D_y-(\frac{N_r}{2}-2)d,\cdots,D_y-d,D_y
		]^T,N_r\;\text{is even};\\
		[&0,\cdots,(\frac{N_r-3}{2})d
		,(\frac{N_r-1}{2})d\;\text{or}\;D_y-(\frac{N_r-1}{2})d,\\
		&D_y-(\frac{N_r-3}{2})d
		,\cdots,D_y-d,D_y
		]^T,N_r \;\text{is odd}.
	\end{aligned}
	\right.\vspace{-0.2cm}
\end{equation}
\end{Theorem}


\begin{proof}
From Lemma 3, we can assert that the first $\lfloor\frac{N_r}{2}\rfloor$ antennas must form a ULA with spacing \( d \) starting from zero to $\left(\lfloor\frac{N_r}{2}\rfloor-1\right)d$. Similarly, based on the symmetry property of $f(\mathbf{y})$ proved in Lemma 2, the last $\lfloor\frac{N_r}{2}\rfloor$ antennas must also form a ULA starting from $D_y-\left(\lfloor\frac{N_r}{2}\rfloor-1\right)d$ to \( D_y \).
Therefore, (\ref{q10})
must hold, which completes the proof.
\end{proof}
 As can be seen from (\ref{q10}), in order to maximize the sensing performance, the MAs should be divided into two groups with equal size, each forming a ULA with antenna spacing \( d \). With given antenna aperture $D_y$, these two groups of MAs should be positioned as far apart from each other as possible, thus maximizing the antenna aperture utilization. It is important to mention that this result is quite intuitive, which resembles how our eyes function, i.e., wider spacing between our eyes enables us to better perceive the object positions in space and also enhances our depth perception and spatial awareness. 
 \vspace{-0.2cm}
 \subsection{CRB Performance Gain Analysis}
 {In our work, we interpret $\theta$ as the direction of interest, i.e., the estimated/predicted target bearing, thus minimizing $\text{CRB}(\theta)$ corresponds directly to optimizing the beamforming vector and APV for that particular azimuth angle.
 	This formulation is typical in the target tracking scenario, where the radar steers its beam toward an estimated/predicted direction to track the movement of the target. Since target motion is continuous, the optimal beamformer and antenna positions only need to vary when there is a sufficient change of the target's angle with the BS.
 	Under this formulation, the full-aperture span $D_y$ of a non-optimized sparse array yields a lower CRB than a half-wavelength ULA with the same $N_r$.  A direct comparison between the non-optimized sparse-array (ULAF) baseline, where antenna spacing is $D_y/(N_r-1)$, and the half-wavelength ULA (ULAH) is given by $			{\text{CRB}_\text{ULAH}}/{\text{CRB}_\text{ULAF}}={f(\mathbf{y}_\text{ULAF})}/{f(\mathbf{y}_\text{ULAH})}
 	={D_y^2}/({(N_r-1)^2 d^2})
 	\ge 1$.
}

Thus, to explore the sensing performance limit brought by the movable receive antennas, we further analyze the ratio between the CRB with receive MAs and the CRB with receive ULAF scheme.
 For ease of analysis and without loss of generality, we assume that the number of receive antennas, i.e., \(N_r\), is even. Then, by substituting (\ref{q10}) and  $\mathbf{y}_\text{ULAF}=[0,\frac{D_y}{N_r-1},\frac{2D_y}{N_r-1},...,D_y]^T$ into \(f(\mathbf{y})\), we obtain\vspace{-0.1cm}
\begin{equation}
	\label{approx1}\small
	\begin{aligned}
	&\frac{\text{CRB}_\text{ULAF}}{\text{CRB}_\text{MA}}=\frac{f(\mathbf{y}_\text{opt})}{f(\mathbf{y}_\text{ULAF})}\\
	&=\frac{N_r-2}{N_r+1}\frac{(N_r-1)d}{D_y}\left(\frac{(N_r-1)d}{D_y}-3\right)
		+3\frac{N_r-1}{N_r+1}
		\\&<3\frac{N_r-1}{N_r+1},
	\end{aligned}\vspace{-0.1cm}
\end{equation}
where the inequality holds due to the fact that the antenna aperture $D_y$ satisfies $D_y \ge (N_r - 1)d$ and ${\text{CRB}_\text{ULAF}}/{\text{CRB}_\text{MA}}$ is a monotonically decreasing function of ${(N_r-1)d}/{D_y}$.
Therefore, the sensing performance gain brought by receive MAs is upper-bounded by \footnote{This upper bound is only achieved when both \( D_y \) and \( N_r \) approach infinity. Since we are considering a far-field channel model, the CRB gain at the receiver is restricted to less than 4.77 dB.}\vspace{-0.2cm}
\begin{equation}\small
	10 \lg\left(\frac{\text{CRB}_{\text{ULAF}}}{\text{CRB}_{\text{MA}}}\right) < 10 \lg 3 = 4.77\,\text{dB}.
\end{equation}

 \vspace{-0.2cm}
\section{Transmit MA Case}
In this section, we focus on the case with transmit MAs only, while the receive antennas are assumed to be a ULA with half wavelength spacing. 
First, we show that the coupling between the steering vector $\mathbf{a}(\mathbf{x})$ and the channel vector $\mathbf{h}(\mathbf{x})$ results in a ``quasi-harmonic" behavior of the objective function (\ref{p2a}). Thus, problem (P) has many local optimal solutions as well as saddle points, thus it is almost impossible to find its global optimal solution using existing algorithms. To address this difficulty, we analyze the LoS channel scenario first and propose a BT-BFS algorithm, which is proved to obtain the global optimal solution of problem (P) in the case with transmit MAs only. Additionally, to reduce the computational complexity of the BT-BFS algorithm, we further introduce a BT-DFS algorithm to obtain a local optimal solution with only linear complexity. Then, for the general NLoS channel scenario, we propose an MM-based RGP algorithm, where a meticulously designed method is presented to determine a high-performance initial point which can effectively prevent the algorithm from converging to a bad stationary point.
Besides, we discover that as the SNR threshold required for communication exceeds a certain value, the sensing CRB will inevitably increase, which cannot be avoided even when the receive antenna aperture is infinitely large. 
However, with movable transmit antennas, we can significantly raise this threshold with the help of our proposed algorithm, allowing the CRB to remain at the lowest level even at higher required communication SNR threshold.
\vspace{-0.1cm}
\subsection{Problem Transformation}
In the case with transmit MAs only, problem (P) can be rewritten as \vspace{-0.2cm}
\begin{subequations}
	\label{p5}\small
	\begin{align}
		  (\text{BP})	\quad \quad\max\limits_{\mathbf{x},\mathbf{w}}\;\;
		&{\begin{vmatrix}{{\mathbf{a}^H}{\mathbf{w}}}\end{vmatrix}^2} &&
		\label{p5a}\\
		\text{s.t.}\;\;&\begin{vmatrix}{{\mathbf{h}^H}{\mathbf{w}}}\end{vmatrix}^2\ge\Gamma\sigma_C^2,&&
		\label{p5b}\\
		&\|
			\mathbf{w}
		\|^2\le{P_T},
		\mathbf{U}\mathbf{x}\preceq\mathbf{l}_u.&&\label{p5c}
	\end{align}
\end{subequations}
Although problem (BP) looks very similar to problem (AP), they are fundamentally different because the receive antennas in (AP) do not participate in the information transmission process, whereas the transmit antenna positions in (BP) will affect the communication channel significantly, which further impacts the transmit beamforming design as well as both communication and sensing performance. Generally, problem (BP) is more challenging to solve than problem (AP) since the transmit APV \(\mathbf{x}\) appears in the exponential parts of the steering vector \(\mathbf{a}\) and the channel vector \(\mathbf{h}\) (see (\ref{s1}) and (\ref{q18})), and \(\mathbf{w}\) and \(\mathbf{x}\) are tightly coupled in both the constraints and objective function. 

To address problem (BP), we first propose to adopt the optimal closed-form solution for \(\mathbf{w}\) in (\ref{q11}), which is able to reduce the number of optimization variables and decouple the variables. Then, through some mathematical manipulations,
we can show that problem (BP) can be equivalently decomposed into the following two subproblems:\vspace{-0.5cm}

\begin{subequations}
	\label{p6}\small
	\begin{align}
	(\text{SP1})	\quad \quad 	\text{Find}&\;\;\mathbf{x}\label{p6a}&&
		\\
		\text{s.t.}&\;\;\begin{vmatrix}{{\mathbf{h}^H}{\mathbf{a}}}\end{vmatrix}^2>\frac{N_t}{P_T}\Gamma\sigma_C^2,\label{p6b}&&\\
		&\;\;\mathbf{U}\mathbf{x}\preceq\mathbf{l}_u,\label{p6c}&&
	\end{align}
\end{subequations}
\vspace{-0.7cm}
\begin{subequations}
	\label{p7}\small
	\begin{align}
\;	(\text{SP2})	\quad \quad\; \max\limits_{\mathbf{x}}&\;\;
		{f_t}{(\mathbf{x})}\label{p7a}&&
		\\
		\text{s.t.}&\;\;\begin{vmatrix}{{\mathbf{h}^H}{\mathbf{a}}}\end{vmatrix}^2\le
		\frac{N_t}{P_T}\Gamma\sigma_C^2,\label{p7b}&&\\ 
		&\quad \text{(\ref{p6c})},&&\nonumber 
	\end{align}
	\vspace{-0.2cm}
\end{subequations}
where \vspace{-0cm}
\begin{equation}
	\label{ft}\small
	\begin{aligned}
{f_t}{(\mathbf{x})}
\triangleq{{\frac{\sqrt{\Gamma\sigma_C^2}}{\|\mathbf{h}\|^2}
		{\begin{vmatrix}{{\mathbf{h}^H}{\mathbf{a}}}\end{vmatrix}}
		+\sqrt{P_T-\frac{\Gamma\sigma_C^2}{\|\mathbf{h}\|^2}}
		\sqrt{\frac{\|\mathbf{h}\|^2\|\mathbf{a}\|^2
				-|{{\mathbf{h}^H}{\mathbf{a}}}|^2}
			{\|\mathbf{h}\|^2}}}}.
	\end{aligned}\vspace{-0.1cm}
\end{equation}
Thus, if subproblem (SP1) has a feasible solution,
it is also the optimal solution of problem (BP); otherwise, the solution of subproblem (SP2) would be the optimal solution.
 
Then, without loss of optimality, subproblem (SP1) can be further transformed into \vspace{-0.1cm}
\begin{equation}
	\label{p8}\small
	\begin{aligned}
		(\text{BP1})	\quad \quad  \max\limits_{\mathbf{x}}\;\;&p_1(\mathbf{x})\triangleq\begin{vmatrix}{{\mathbf{h}^H}{\mathbf{a}}}\end{vmatrix}^2
		&&\\
		\text{s.t.}\;\;&\text{(\ref{p6b})},\text{(\ref{p6c})}.&&
	\end{aligned}\vspace{-0.1cm}
\end{equation}
Next, we resort to the following trigonometric transformations to simplify subproblem (SP2):\vspace{-0.2cm}
\begin{equation}
	\label{q12}\small
	\begin{aligned}
		\cos\upsilon(\mathbf{x})={{\begin{vmatrix}{{\mathbf{h}^H}{\mathbf{a}}}\end{vmatrix}}}/
		{({\|{\mathbf{h}}\|}\|\mathbf{a}\|)},\upsilon(\mathbf{x})\in[0,\frac{\pi}{2}],
	\end{aligned}
\end{equation}
\vspace{-0.4cm}
\begin{equation}
	\label{q13}\small
	\begin{aligned}
		\sin\phi(\mathbf{x})=
		\sqrt{{\Gamma\sigma_C^2}/{(P_T\|\mathbf{h}\|^2})}
		,\phi(\mathbf{x})\in[0,\frac{\pi}{2}],
	\end{aligned}\vspace{-0.2cm}
\end{equation}
based on which, $f_t(\mathbf{x})$ can be equivalently rewritten as \vspace{-0.1cm}
\begin{equation}
	\label{q20}\small
	\begin{aligned}
		f_t(\mathbf{x})&\overset{(a)}{=}\sqrt{N_t}\bigg(\frac{\sqrt{\Gamma\sigma_C^2}}{\|\mathbf{h}\|}
		\cos\upsilon
		+\sqrt{P_T-\frac{\Gamma\sigma_C^2}{\|\mathbf{h}\|^2}}\sin\upsilon\bigg)
		\\&=\sqrt{N_tP_T}\sin\left(\upsilon(\mathbf{x})+\phi(\mathbf{x})\right),
	\end{aligned}\vspace{-0.1cm}
\end{equation}
where $(a)$ holds due to the fact that $\|\mathbf{a}\|=\sqrt{N_t}$. Besides, by substituting (\ref{q12}) and (\ref{q13}) into $|{{\mathbf{h}^H}{\mathbf{a}}}|^2\le
\frac{N_t}{P_T}\Gamma\sigma_C^2$, we obtain\vspace{-0.1cm}
\begin{equation}
	\label{q15}\small
	 \phi(\mathbf{x})+\upsilon(\mathbf{x})\ge\frac{\pi}{2}.
\vspace{-0.1cm}
\end{equation}
Therefore, subproblem (SP2) can be equivalently recast as \vspace{-0.1cm}
\begin{equation}
	\label{p9}\small
	\begin{aligned}
	(\text{BP2})	\quad \quad 	\max\limits_{\mathbf{x}}\;\;&\sqrt{N_tP_T}\sin(\upsilon(\mathbf{x})+\phi(\mathbf{x}))\\
		\text{s.t.}\;\;&\text{(\ref{q15})}
		,\text{(\ref{p6c})}.
	\end{aligned}\vspace{-0.1cm}
\end{equation}
The advantages of adopting these trigonometric transformations will be explained later.

After the above decomposition and transformations, we observe that although the subproblems only contain a single variable \(\mathbf{x}\), they are still challenging to solve due to the term $|{{\mathbf{h}^H}{\mathbf{a}}}|$. Specifically, when expanding $|{{\mathbf{h}^H}{\mathbf{a}}}|$ using Euler's formula, it becomes the sum of multiple sine and cosine functions with different periods, exhibiting a harmonic-like characteristic. Thus, it is almost impossible to find the global optimal solutions of problems (BP1) and (BP2) using existing algorithms. Fortunately, when the channel between the BS and the user follows the LoS model, i.e., \(L_t = 1\), these two subproblems can be further simplified, and the corresponding optimal solutions can be obtained. Therefore, in the following, we will consider the LoS channel scenario first and then address the more general NLoS channel scenario.
\vspace{-0.3cm}
\subsection{LoS Channel Scenario}
In this scenario, the channel vector in (\ref{q18}) can be re-expressed as \vspace{-0.2cm}
\begin{equation}
	\label{q19}\small
	\mathbf{h}^H(\mathbf{x})=\sigma_1^*\begin{bmatrix}e^{-j\frac{2\pi}{\lambda}x_1\sin\theta_t^1},...,e^{-j\frac{2\pi}{\lambda}x_{N_t}\sin\theta_t^1}\end{bmatrix}.
\end{equation}
Thus, $\sin\phi(\mathbf{x})$ in (\ref{q13})  becomes a constant $\sqrt{\frac{\Gamma\sigma_C^2}{P_T|\sigma_1|^2N_t}}$ that is independent of $\mathbf{x}$. Note that this observation can be easily obtained after converting $f_t(\mathbf{x})$ into $\sqrt{N_tP_T}\sin(\upsilon+\phi)$ via the trigonometric transformations introduced in (\ref{q12}) and (\ref{q13}), whereas the same does not hold true for the original subproblem (SP2). 
Based on this key observation {and constraint (\ref{q15}), together with (\ref{q12}) and (\ref{q13}), we can see that the argument $\upsilon(\mathbf x)+\phi$ lies within $[\pi/2,\pi]$, where $\sin(\cdot)$ is strictly decreasing. Thus, maximizing $\sin(\upsilon(\mathbf{x})+\phi)$ is equivalent to minimizing $\upsilon(\mathbf{x})$. Since $\upsilon(\mathbf{x})\in[0,\pi/2]$ and $\cos \upsilon(\mathbf{x})$ is strictly decreasing on $\upsilon(\mathbf{x})\in[\pi/2,\pi]$, minimizing $\upsilon(\mathbf{x})$ is also equivalent to maximizing $\cos\upsilon(\mathbf{x})$. Therefore, problem (BP2) can be further converted into the following more concise form without loss of optimality:}\vspace{-0cm}
\begin{equation}
	\label{p10}\small
	\begin{aligned}
		\max\limits_{\mathbf{x}}\;\;&\sqrt{N_tP_T}\cos\upsilon(\mathbf{x})\\
		\text{s.t.}\;\;&\text{(\ref{q15})},\text{(\ref{p6c})}.
	\end{aligned}\vspace{-0.1cm}
\end{equation}
Furthermore, for ease of analysis, we substitute (\ref{q12}) back into subproblem (\ref{p10}) and obtain \vspace{-0.1cm}
\begin{equation}
	\label{p11}\small
	\begin{aligned}
		\max\limits_{\mathbf{x}}\;\;&\sqrt{{P_T}/{(N_t\begin{vmatrix}\sigma_1\end{vmatrix}^2)}}
		\begin{vmatrix}\mathbf{h}^H\mathbf{a}\end{vmatrix}\\
		\text{s.t.}\;\;&\text{(\ref{p7b})},\text{(\ref{p6c})}.
	\end{aligned}\vspace{-0.1cm}
\end{equation}
To this end, the subproblems (BP1) and (\ref{p11}) (equivalent to problem (BP2)) can be recombined into a simpler problem under the LoS scenario, which is given by\vspace{-0.1cm}
\begin{equation}
	\label{p12}\small
		 \begin{aligned}
		 	(\text{CP})	\quad \quad 
		\max\limits_{\mathbf{x}}\;\;&g(\mathbf{x})\triangleq
		\begin{vmatrix}\sum\limits_{i=1}^{N_t}e^{-j\frac{2\pi}{\lambda}(\sin\theta_t^1+\sin\theta)x_i}\end{vmatrix}\\
		\text{s.t.}\;\;&\text{(\ref{p6c})}.
	\end{aligned}\vspace{-0.1cm}
\end{equation}
Problem (CP) contains \(N_t\) linear inequalities and its objective function is the modulus of the sum of \(N_t\) unit-magnitude complex numbers, i.e., $e^{-j\frac{2\pi}{\lambda}(\sin\theta_t^1+\sin\theta)x_i}$, $i=1,2,...,N_t$. In the real domain, this translates to the sum of \(N_t\) sine functions with different periods, which introduces strong periodicity in each dimension of \(\mathbf{x}\). 
In the following, we analyze and solve problem (CP) under different values of \(D_x\).

First, when $D_x$ is sufficiently large and satisfies $D_x\ge\frac{(N_t-1)\lambda}{\sin\theta_t^1+\sin\theta}$, aligning the $N_t$ complex numbers along the same line in the complex plane yields the optimal solution of problem (CP), which is given by \vspace{-0.1cm}
\begin{equation}
	\label{q17}\small
	\begin{aligned}
		\mathbf{x}_\text{opt}=\left[
			0,  \frac{\lambda}{\sin\theta_t^1+\sin\theta} ,\cdots , \frac{(N_t-1)\lambda}{\sin\theta_t^1+\sin\theta}\right]^T.
	\end{aligned}\vspace{-0.1cm}
\end{equation}
Although the above optimal antenna positions are equidistant as in a conventional ULA, the antenna spacing becomes \(\frac{\lambda}{\sin\theta_t^1 + \sin\theta}\), instead of \(\frac{\lambda}{2}\).\footnote{This implies that when the azimuth AoD \(\theta_t^1\) of the LoS channel transmit path and the target angle \(\theta\) are relatively small, the antenna aperture \(D_x\) needs to be very large in order to satisfy (\ref{q17}).}

Second, when  $D_x<\frac{(N_t-1)\lambda}{\sin\theta_t^1+\sin\theta}$ (note that $D_x$ should satisfy $D_x\ge(N_t-1)d$), problem (CP) becomes more challenging to solve due to the insufficient space to align the \(N_t\) complex numbers. 
To proceed, we first consider a general function defined as ${E(\bm{\beta})=|\sum\nolimits_{i=1}^{N}\rho_ie^{j\beta_i}|}$, where $
\bm{\beta}\triangleq [\beta_1, \cdots, \beta_N]^T \in\mathbb{R}^{N\times1}$ and $\rho_i \geq 0, \forall i$, then a useful property of $E(\bm{\beta})$ can be exploited to further simplify problem (CP), which is given as follows.
\begin{Lemma}
$E(\bm{\beta})$ has a unique local maximum value of $\sum\nolimits_{i=1}^{N}\rho_i$, and to attain this local maximum, $\bm{\beta}$ must satisfy \vspace{-0.1cm}
	\begin{equation}
		\label{Econ}\small
		\begin{aligned}
			\beta_m-\beta_n=2\pi{k},1\le{m},{n}\le{N},k\in\mathbb{Z}.
		\end{aligned}\vspace{-0.1cm}
	\end{equation}
\end{Lemma}

\begin{proof}
	See Appendix C.
\end{proof}

Lemma 4 indicates that \( E(\bm{\beta}) \) only has one local maximum value which is also the global maximum value.
Besides, it can be observed that the objective function of problem (CP), i.e.,  \( g(\mathbf{x}) \), is a special case of \( E(\bm{\beta}) \) by setting $\bm{\rho} \triangleq [\rho_1,\cdots, \rho_N]^T=\mathbf{1}$ and $\bm{\beta}={-\frac{2\pi}{\lambda}(\sin\theta_t^1+\sin\theta)}\mathbf{x}$.
Since the feasible region of problem (CP) is a compact set and \( g(\mathbf{x}) \) is a continuous function, its maximum must be achieved either at a local maximum point or on the boundary.
According to Lemma 4,  when $D_x\ge\frac{(N_t-1)\lambda}{\sin\theta_t^1+\sin\theta}$, satisfying condition (\ref{Econ}) actually implies that (\ref{q17}) is fulfilled. Therefore, when $D_x<\frac{(N_t-1)\lambda}{\sin\theta_t^1+\sin\theta}$, the optimal solution of problem (CP) must lie on the boundary
, which means that some of the linear inequalities in $\mathbf{U}\mathbf{x}\preceq\mathbf{l}_u$ must be active.
We assume that the number of active constraints is \( c \), with \( 1 \leq c \leq N_t-1 \).\footnote{If \( c=N_t \), then \( \mathbf{U}\mathbf{x} = \mathbf{l}_u \) and \( D_x = (N_t - 1)\lambda/2 \) hold, which implies that the transmit antennas should be arranged at equal spacing with an interval of $\lambda/2$. This is equivalent to the conventional ULA and thus we can safely ignore the case of $c=N_t$.}

Based on the above analysis, it follows that to obtain the optimal solution of problem (CP), we can check the corresponding Karush-Kuhn-Tucker (KKT) conditions. To describe the optimality conditions for a nonlinear optimization problem such as problem (CP), it is often required to assume that some regularity conditions are met \cite{book1}. Here, we employ a commonly used condition, known as the linear independence constraint qualification (LICQ), and the detailed verification is provided in Appendix D.

%

Now we are ready to establish the KKT conditions of problem (CP).
Suppose that \( \mathbf{x}^* \) is a local maximum point of (CP), and the LICQ holds at \( \mathbf{x}^* \), then there exists a Lagrange multiplier 
vector \( \bm{\lambda}^*=[\lambda_1^*, \lambda_2^*, \cdots, \lambda_{N_t}^*]^T\), such that the following conditions are satisfied at 
\( (\mathbf{x}^*, \bm{\lambda}^*) \):\vspace{-0.4cm}

\begin{subequations}
	\label{kkt}\small
	\begin{align}
		&\nabla_{\mathbf{x}} \mathcal{L}(\mathbf{x}^*, \bm{\lambda}^*) = \mathbf{0}, \label{kkt1}\\
		&\mathbf{U}\mathbf{\mathbf{x}^*}-\mathbf{l}_u\preceq\mathbf{0},\label{kkt2}\\
		&\bm{\lambda}^*\succeq\mathbf{0},\label{kkt3}\\
		&\bm{\lambda}^*\circ(\mathbf{U}\mathbf{\mathbf{x}^*}-\mathbf{l}_u)=\mathbf{0},\label{kkt4}
	\end{align}
\end{subequations}
where $\mathcal{L}(\mathbf{x}, \bm{\lambda})$ is the Lagrangian function defined as $	 	\mathcal{L}(\mathbf{x}, \bm{\lambda})=g(\mathbf{x})+
\bm{\lambda}^T(\mathbf{U}\mathbf{\mathbf{x}}-\mathbf{l}_u)$.
It is important to mention that solutions that satisfy the KKT conditions are not necessarily local optimum, but any local optimum must satisfy the KKT conditions. 
Furthermore, due to the highly non-convexity of $\nabla_{\mathbf{x}} \mathcal{L}(\mathbf{x}^*, \bm{\lambda}^*)$, it is generally impossible to find all the points that satisfy the KKT conditions, which means that obtaining the global optimal solution of problem (CP) is very difficult.
To tackle this difficulty, we will first present the proposed BT-BFS algorithm in the following to identify
all the local optimum of problem (CP) by utilizing the special properties of $E(\bm{\beta})$ given in Lemma 4. Then, we will provide the detailed proof for that the proposed algorithm is guaranteed to find the global optimal solution.

Without loss of generality, 
let $\mathcal{A}_c \triangleq \{ a_i|\mathbf{U}_{a_i}\mathbf{x} = ({\mathbf{l}_u})_{a_i},1\le{i}\le{c},a_1<a_2<\cdots<a_c \}$ denote the set of active constraint indices, where  
\(\mathbf{U}_{a_i}\) and \((\mathbf{l}_u)_{a_i}\) refer to the \(a_i\)-th row of \(\mathbf{U}\) and the \(a_i\)-th entry of \(\mathbf{l}_u\), respectively.
Due to the \textit{particularity} of the \(N_t\)-th constraint, i.e., $x_{N_t} -x_1 \le D_x$,  we propose to analyze problem (CP) by considering the following two cases.

\textbf{Case}\;\textbf{\uppercase\expandafter{\romannumeral1}}:
 \( a_c \leq N_t - 1 \). 
In this case, by substituting the \( c \) active constraints into $g(\mathbf{x})$, i.e., retaining $x_{a_i},1\le{i}\le{c-1}$ as free variables and substituting \( x_{a_i+1} = x_{a_i} + d \) into $g(\mathbf{x})$, the dimension of the original problem (CP) can be reduced from \(N_t\) to \(N_t - c\). 
After relabeling the remaining \( N_t - c \) variables to form a new vector variable $\mathbf{x}'$, problem (CP) can be reduced to \vspace{-0.1cm}
\begin{equation}
	\label{pcase1}\small
	\begin{aligned}
	(\text{CP-m})	\quad  \max\limits_{\mathbf{x}'}\;\;&{\tilde{g}}(\mathbf{x}')\triangleq
		\begin{vmatrix}
			\sum\limits_{i=1}^{N_t-c}r_ie^{-j\frac{2\pi}{\lambda}(\sin\theta_t^1+\sin\theta)x'_i}
		\end{vmatrix}\\
		\text{s.t.}\;\;&\mathbf{U}'\mathbf{x}'\preceq\mathbf{l}'_{u}.
	\end{aligned}\vspace{-0.2cm}
\end{equation}
where \vspace{-0.2cm}
\begin{equation}
	\small
	\begin{aligned}
		r_i &= \sum\limits_{p=0}^{n_i}
		e^{-j\frac{2\pi}{\lambda}(\sin\theta^1_t+\sin\theta)pd}, \\
		\mathbf{U}' &= \left[\begin{array} {cccccc}
			1 & -1 & 0 & \cdots & 0 & 0\\
			0 & 1 & -1  & \cdots & 0 & 0\\
			\vdots & \vdots  & \vdots & \ddots & \vdots & \vdots\\
			0 & 0 & 0 &  \cdots & 1 & -1\\
			-1 & 0 & 0 &  \cdots & 0 & 1
		\end{array}\right]_{(N_t-c)\times{(N_t-c)}}, \\
		\mathbf{l}'_u &= \left[
		-(n_1+1)d,  -(n_2+1)d ,\cdots , D_x-n_{N_t-c}d 
		\right]^T,
	\end{aligned}\vspace{-0.1cm}
\end{equation}
with $n_i$ denoting the number of active constraints related to $x'_i$. 
Let $\mathbf{x}^{\prime *}$ denote  the optimal solution of problem (CP-m). {Then, by setting $\rho_i = |r_i| \geq 0$ for all $i$, and introducing a phase offset $\bm{\zeta} = \angle \mathbf{r}$ in the argument $\bm{\beta} = -\frac{2\pi}{\lambda}(\sin\theta_t^1 + \sin\theta)\mathbf{x}' + \bm{\zeta}$, ${\tilde{g}}(\mathbf{x}')$ can be viewed as a special case of $E(\bm{\beta})$. Therefore, \( {\mathbf{x}{'}}^* \) must satisfy condition (\ref{Econ}) in Lemma 4, which is given as}\vspace{-0.1cm}
\begin{equation}
	\label{criteria}\small
	\begin{aligned}
		\frac{2\pi}{\lambda}(\sin\theta_t^1+\sin\theta)({x}_{i+1}^{\prime*}-{x}_{i}^{\prime*}){-(\angle{r_{i+1}}-\angle{r_{i}})}=2\pi{k_i},
	\end{aligned}\vspace{-0.1cm}
\end{equation}
where $1\le{i}\le{N_t-c-1},{k_i}\in\mathbb{Z}$.
Additionally, by combining (\ref{criteria}) with the constraint of problem (CP-m), i.e., ${x}_{i+1}^{\prime*}-{x}_{i}^{\prime*}\ge (n_i+1)d$, we can see that $k_i$ must satisfy \vspace{-0.1cm}
\begin{equation}
	\label{k_i_ineq}\small
	\begin{aligned}
		\frac{{\lambda(2\pi{k_i}+\angle{r_{i+1}}-\angle{r_{i}})}}{2\pi(\sin\theta_t^1+\sin\theta)}\ge{(n_i+1)d}.
	\end{aligned}\vspace{-0.1cm}
\end{equation}
{Consequently, the remaining $N_t - c$ inactive constraints form an unbounded set. Thus, we can denote $D_\text{min}={{{x}_{N_t-c}^{\prime{*}}-{x}_{1}^{\prime*}}+n_{N_t-c}d}$ as the smallest \( D_x \) to align the $N_t-c$ complex numbers $r_ie^{-j\frac{2\pi}{\lambda}(\sin\theta_t^1+\sin\theta)x_i^{\prime}},i=1,2,\cdots,N_t-c$.}
Then, by substituting (\ref{criteria}) into $D_\text{min}$, we have \vspace{-0.1cm}
	\begin{align}
		\label{Dmin}\small
		D_\text{min}=&\sum\limits_{i=1}^{N_t-c-1}({x}_{i+1}^{\prime*}-{x}_{i}^{\prime*})+n_{N_t-c}d
		\\=&
		\frac{\lambda}{\sin\theta_t^1+\sin\theta}\sum\limits_{i=1}^{N_t-c-1}{\!\!\!\!}\left( {k_i}+
		\frac{{\angle{r_{i+1}}-\angle{r_{i}}}}{2\pi}\right)+n_{N_t-c}d.\nonumber
	\end{align}
To determine the minimum value of \( D_x \) that satisfies (\ref{criteria}), we minimize \( D_{\text{min}} \), thereby obtaining the minimum \( k_i \) from (\ref{k_i_ineq}) as\vspace{-0.0cm}
\begin{equation}
	\label{k_i}\small
	\begin{aligned}
		k_i=\left\lceil{
			\frac{\sin\theta_t^1+\sin\theta}{\lambda}(n_i+1)d-\frac{{\angle{r_{i+1}}-\angle{r_{i}}}}{2\pi}
			}\right\rceil.
	\end{aligned}\vspace{-0.1cm}
\end{equation}

\textbf{Case}\;\textbf{\uppercase\expandafter{\romannumeral2}}: \( a_c = N_t \). 
Now we consider the case when the $N_t$-th constraint is active, i.e., \(x_{N_t} - x_1 = D_x\). 
Let \(l\) denote the smallest index of $\mathbf{x}$ such that the $(l+1)$-th to the $N_t$-th constraints are active.
Similar to the previous case, by substituting the \( c \) active constraints into $g(\mathbf{x})$, i.e., retaining $x_{a_i},1\le{i}\le{c-(N_t-l)+1}$ and $x_1$ as free variables and substituting \( x_{a_i+1} = x_{a_i} + d \) and $x_{N_t-j}=x_1+D_x-jd,0\le{j}\le{N_t-l-1}$ into $g(\mathbf{x})$, 
we can obtain a similar problem as (CP-m), 
with\vspace{-0.1cm}
\begin{equation}\small
	\begin{aligned}
		{r_i}&  {=
     \begin{cases}
	\sum\limits_{p=0}^{n_1+l-N_t}
	e^{-j\frac{2\pi}{\lambda}(\sin\theta^1_t+\sin\theta)pd}\\
	\quad \quad + \sum\limits_{p=0}^{N_t-l-1}
	e^{j\frac{2\pi}{\lambda}(\sin\theta^1_t+\sin\theta)(pd-D_x)}, \quad i = 1,\\
	\sum\limits_{p=0}^{n_i}
	e^{-j\frac{2\pi}{\lambda}(\sin\theta^1_t+\sin\theta)pd}, \quad \quad\quad\; 2 \le i \le N_t-c,\\
    \end{cases}}\\
		\mathbf{U}' &= 
		\left[\begin{array} {ccccccc}
			1 & -1 & 0 & 0 & \cdots & 0 & 0\\
			0 & 1 & -1 & 0 & \cdots & 0 & 0\\
			\vdots & \vdots & \vdots & \vdots & \ddots & \vdots & \vdots\\
			0 & 0 & 0 & 0 & \cdots & 1 & -1
		\end{array}\right]_{(N_t-c)\times{(N_t-c)}},\\
		\mathbf{l}'_u &=
		\left[
		-(n_1+1)d,  -(n_2+1)d ,\cdots , -(n_{N_t-c}+1)d
		\right]^T.
	\end{aligned}
\end{equation}
{Here, the $N_t - c$ inactive constraints form a bounded set, which can be intuitively viewed as a closed chain. 
	To handle this chain, the boundary is “cut” at the last inactive constraint not related to the 
	$N_t$-th constraint, which can effectively open the chain and convert it into a similar unbounded set as in Case I.}
Due to the \textit{particularity} of the \(N_t\)-th constraint, we re-denote \(D_{\text{min}}\) in this case as follows:\vspace{-0.1cm}
\begin{align}
	\label{Dmin_2}\small
	D_\text{min}=&{{x}_{N_t-c}^{\prime*}-{x}_{1}^{\prime*}}+({N_t-l})d
	\\=&
	\frac{\lambda}{\sin\theta_t^1+\sin\theta}{\!\!\!\!}\sum\limits_{i=1}^{N_t-c-1}{\!\!\!\!}\left( {k_i}  + 
	\frac{{\angle{r_{i+1}}-\angle{r_{i}}}}{2\pi}\right)\!+\!({N_t-l})d,\nonumber
\end{align}
and the minimum $k_i$ to minimize $D_{\text{min}}$ is also given by (\ref{k_i}).
{Thus, the underlying factor for distinguishing the two cases is whether the last constraint is active or not, which determines if the inactive constraints form a bounded or unbounded set. This geometric difference directly results in the distinct expressions of $D_{\min}$ in (\ref{Dmin}) and (\ref{Dmin_2}). }

Therefore, we can see that with given $c$ and $a_c$, there exists a constant $D_{\text{min}}$ such that when $D_{\text{min}}\le D_x <\frac{(N_t-1)\lambda}{\sin\theta_t^1+\sin\theta} $, the local optimal solution of problem (CP) can also be easily obtained without exploiting the KKT conditions in (\ref{kkt}). Specifically, for both cases of $a_c \le N_t-1$ and $a_c=N_t$, the optimal solution of problem (CP-m) is as follows if \(D_x \geq D_{\text{min}}\):\vspace{-0.1cm}
\begin{equation}
	\label{pcase1opt}\small
	\begin{aligned}
		{x}_{i+1}^{\prime*}={x}_{i}^{\prime*}+\frac{k_i\lambda}{\sin\theta_t^1+\sin\theta}+\frac{{\lambda(\angle{r_{i+1}}-\angle{r_{i}})}}{2\pi(\sin\theta_t^1+\sin\theta)},
	\end{aligned}\vspace{-0.1cm}
\end{equation}
and \({x}_{1}^{\prime*}\) can be assumed to be 0 without loss of generality (in fact, (\ref{q17}) is a special case of (\ref{pcase1opt}) when $c=0$). When \( D_x < D_{\text{min}} \), however, the optimal solution of problem (CP-m) must lie on the boundary of its feasible region, which means that it can only be attained by considering all possible combinations of \(a_1, a_2, \dots, a_c\) and the value of $c$ should be enumerated from $1$ to $N_t-1$, such that  
all local optimal values of problem (CP) are obtained. 
In this way, the most challenging part of the KKT conditions, i.e., (\ref{kkt1}), is automatically satisfied, {{where the detailed proof is given in Appendix E}}. Thus, by comparing these local optimal values, the global optimal solution can be determined.

\begin{figure}[t]
	\centering
	\includegraphics[width=3.2in]{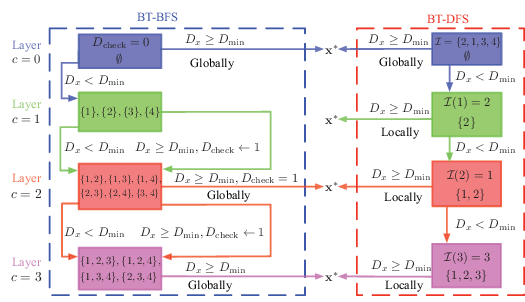}
	\caption{Illustration of the BT-BFS and BT-DFS algorithms when \( N_t = 4 \).}
	\label{BFSDFS}
	\vspace{-0.0cm}
\end{figure}
To traverse all the boundaries of the feasible region (\ref{p6c}), $2^{N_t}$ cases are required to be evaluated which results in extremely high computational complexity. 
To lower the computational complexity and in the meantime still achieve the global optimal solution of problem (CP), we propose a BFS strategy, as shown in Fig. \ref{BFSDFS} (for a toy example with $N_t=4$), where the $c$-th ``layer" represents that $c$ constraints are active and \( \binom{N_t}{c} \) combinations of possible active constraints should be considered in this layer. These layers are traversed in a breadth-first manner from $c=0$.
In the example given in Fig. \ref{BFSDFS}, the overall search process is divided into four layers, i.e., \( c \in \{ 0, 1, 2, 3\} \), and $c$ also represents the size of the active set \( \mathcal{A}_c \). In each layer, there are a total number of \( \binom{4}{c} \) subsets. For each subset, if \( D_x \geq D_{\text{min}} \), a local optimum solution of problem (CP) can be easily found through (\ref{pcase1opt}).
On the other hand, if \( D_x < D_{\text{min}} \), the proposed algorithm will continue searching downward to the next layer until \( D_x \geq D_{\text{min}} \) is satisfied.
Generally, the solutions obtained from the upper layers are better than those from the lower layers, as fewer active constraints are imposed in the upper layers and the corresponding objective function is supposed to have a higher degree of freedom.
However, since the \(N_t\)-th constraint is special as it limits the overall antenna size, it is possible that when the \(N_t\)-th constraint is active, the solutions obtained under fewer active constraints lead to worse outcomes. Therefore, in such cases, we need to perform an additional layer of traversal to ensure that the solution obtained after BFS is indeed optimal.
The detailed steps of the proposed BT-BFS algorithm to solve problem (CP) are summarized in Algorithm $\ref{A2}$. 
\begin{algorithm}[t]
	\caption{Proposed BT-BFS Algorithm}
	\label{A2}
	\renewcommand\baselinestretch{\linspreadalgr}\selectfont
	\renewcommand{\algorithmicrequire}{\textbf{Input:}}
	\renewcommand{\algorithmicensure}{\textbf{Initialize:}}
	\begin{algorithmic}[1]
		\REQUIRE $N_t$, $D_x$, $d$, $\theta$, $\theta_t^1$,$\lambda$.  
		\ENSURE $L\leftarrow0$, $c\leftarrow1$,  $g_\text{max}\leftarrow0$, $D_\text{check}\leftarrow0$.

		\STATE \textbf{if} {$D_x\ge\frac{(N_t-1)\lambda}{\sin\theta_t^1+\sin\theta}$} \textbf{then} obtain $\mathbf{x}_\text{opt}$ via (\ref{q17}).
		\STATE \textbf{else} \textbf{while} {$c\le{N_t}-1$} \textbf{do} 
		\STATE \;\;\;\;\;\;\;\;\;\;\;$j\leftarrow1$.
		\STATE \;\;\;\;\;\;\;\;\;\;\;\textbf{while} {$j\le{\binom{N_t}{c}}$} \textbf{do}
		\STATE \;\;\;\;\;\;\;\;\;\;\;\;\;\;\;Obtain $\mathcal{A}_c$ and $\{n_i\}_{i=1}^{N_t-c}$. 
		\STATE \;\;\;\;\;\;\;\;\;\;\;\;\;\;\;\textbf{if} {$a_c<N_t$} \textbf{then} compute $D_\text{min}$ via (\ref{Dmin}).
		\STATE \;\;\;\;\;\;\;\;\;\;\;\;\;\;\;\textbf{else} Compute $D_\text{min}$ via (\ref{Dmin_2}). \textbf{end if}
		\STATE \;\;\;\;\;\;\;\;\;\;\;\;\;\;\;Obtain $\mathbf{x}'$ via (\ref{pcase1opt}), $j\leftarrow{j+1}$.
		\STATE \;\;\;\;\;\;\;\;\;\;\;\;\;\;\;\textbf{if} {$D_\text{min}\le{D_x}$ and $\tilde{g}_\text{max}<{\tilde{g}}(\mathbf{x}')$} \textbf{then} 
		\STATE \;\;\;\;\;\;\;\;\;\;\;\;\;\;\;$\mathbf{x}'_\text{opt}\leftarrow\mathbf{x}'$, $\tilde{g}_\text{max}\leftarrow{\tilde{g}(\mathbf{x}')}$, ${L}\leftarrow{c}.$ \textbf{end if}
		\STATE \;\;\;\;\;\;\;\;\;\;\;\textbf{end while}
		\STATE \;\;\;\;\;\;\;\;\;\;\;\textbf{if} {$D_\text{check}=1$}  \textbf{then} \textbf{break}.
		\STATE \;\;\;\;\;\;\;\;\;\;\;\textbf{else if} {$L>0$} \textbf{then} $c\leftarrow{L}+1$, $D_\text{check}\leftarrow1$.
		\STATE \;\;\;\;\;\;\;\;\;\;\;\textbf{else} {$c\leftarrow{c+1}$.} \textbf{end if}
		\STATE \;\;\;\;\;\;\;\textbf{end while} 
		\STATE \textbf{end if}
		\RETURN $\mathbf{x}_\text{opt}$ obtained from $\mathbf{x}'_\text{opt}$.
	\end{algorithmic}
	\vspace{-0.1cm}
\end{algorithm}

\begin{algorithm}[t]
	\caption{Proposed BT-DFS Algorithm}
	\label{A3}
	\renewcommand\baselinestretch{\linspreadalgr}\selectfont
	\renewcommand{\algorithmicrequire}{\textbf{Input:}}
	\renewcommand{\algorithmicensure}{\textbf{Initialize:}}
	\begin{algorithmic}[1]
		\REQUIRE $N_t$, $D_x$, $d$, $\theta$, $\theta_t^1$,$\lambda$.  
		\ENSURE $\mathcal{I}\leftarrow{\text{a random sequence from $1$ to $N_t$}}$, $c\leftarrow1$.
		\STATE \textbf{if} {$D_x\ge\frac{(N_t-1)\lambda}{\sin\theta_t^1+\sin\theta}$} \textbf{then} obtain $\mathbf{x}_\text{opt}$ via (\ref{q17}).
		\STATE \textbf{else} \textbf{while} {$c\le{N_t}-1$} \textbf{do} 		
		\STATE \;\;\;\;\;\;\;\;\;\;\;Add ${\mathcal{I}(c)}$ into $\mathcal{A}_c$ and obtain $\{n_i\}_{i=1}^{N_t-c}$. 
		\STATE \;\;\;\;\;\;\;\;\;\;\;\textbf{if} {$a_c<N_t$} \textbf{then} compute $D_\text{min}$ via (\ref{Dmin}).
		\STATE \;\;\;\;\;\;\;\;\;\;\;\textbf{else} Compute $D_\text{min}$ via (\ref{Dmin_2}). \textbf{end if}
		\STATE \;\;\;\;\;\;\;\;\;\;\;Obtain $\mathbf{x}'$ via (\ref{pcase1opt}), $c\leftarrow{c+1}$.
		\STATE \;\;\;\;\;\;\;\;\;\;\;\textbf{if} $D_\text{min}\le{D_x}$ \textbf{then} 
		$\mathbf{x}'_\text{local}\leftarrow\mathbf{x}',$ \textbf{break}. \textbf{end if}
		\STATE \;\;\;\;\;\;\;\textbf{end while} 
		\STATE\textbf{end if}
		\RETURN $\mathbf{x}_\text{local}$ obtained from $\mathbf{x}'_\text{local}$.
	\end{algorithmic}\vspace{-0.1cm}
\end{algorithm}

To further reduce the computational complexity, we propose in the following the BT-DFS algorithm. 
Different from the BT-BFS algorithm, the BT-DFS algorithm adds constraints to the active set \(\mathcal{A}_c\) one by one until \(D_x \geq D_{\text{min}}\) is achieved.
As shown in the example in Fig. \ref{BFSDFS} (b), a random sequence \( \mathcal{I} = \{2, 1, 3, 4\} \) is chosen to represent the searching path of the proposed depth-first traversal, which means that the 2nd, 1st, 3rd, and 4th constraints are sequentially activated in each layer until \( D_x \geq D_{\text{min}} \) holds.
Although this algorithm does not guarantee the optimal solution for problem (CP) generally, it can achieve a local optimal solution with reasonably good performance, and the corresponding details are provided in Algorithm \ref{A3}.

\begin{remark}
	{The proposed BT-BFS algorithm fundamentally differs from the conventional active-set method by systematically enumerating all active constraint combinations and computing each boundary’s unique maximizer in closed form, rather than iteratively solving quadratic programming (QP) subproblems} 
	{and updating the working set. Leveraging the problem’s quasi-harmonic structure, BT-BFS directly satisfies KKT conditions on each boundary without convex approximations or multiplier updates. Its breadth-first search with immediate feasibility checks guarantees global optimality, whereas the active-set method generally converges only to a stationary point in nonconvex settings. Although BT-BFS has worst-case complexity $O(2^{N_t})$, each boundary evaluation is $O(1)$, and it generally converges rapidly in simulations. A lower-complexity variant, BT-DFS, further reduces complexity to $O(N_t)$ by following a single depth-first path, which offers significant computational saving compared to the repeated $O(N_t^3)$ QP solvers in the active-set approach.}
\end{remark}


\vspace{-0.2cm}
\subsection{NLoS Channel Scenario}
For the general NLoS channel scenario, the channel vector $\mathbf{h}$ becomes more complex, and it is not feasible to recombine the two subproblems (BP1) and (BP2) into a single problem, as in Section \uppercase\expandafter{\romannumeral4}. B. Therefore, we propose an MM-based RGP algorithm to solve these two subproblems.

First, we focus on subproblem (BP1). 
As can be seen, its objective function, i.e., $p_1(\mathbf{x})$, can be equivalently transformed into
$p_1(\mathbf{x})={\bm{\psi}(\mathbf{x})}^H\bm{\Sigma}\bm{\psi}(\mathbf{x})$, 
where $\bm{\Sigma}\triangleq\bm{\sigma}{\bm{\sigma}}^H$ is a positive semi-definite matrix and 
$\bm{\psi}(\mathbf{x})\triangleq\mathbf{G}_{L_t\times{N_t}}\mathbf{a}_{N_t\times{1}}=[
\sum\nolimits_{i=1}^{N_t}e^{-j\frac{2\pi}{\lambda}(\sin\theta_t^1+\sin\theta)x_i},\cdots,
\sum\nolimits_{i=1}^{N_t}e^{-j\frac{2\pi}{\lambda}(\sin\theta_t^{L_t}+\sin\theta)x_{i}} ]^T.$
Since \( p_1(\mathbf{x}) \) is convex with respect to (w.r.t.) \( {\bm{\psi}(\mathbf{x})} \), we can derive the following lower bound of \( {\bm{\psi}(\mathbf{x})}^H\bm{\Sigma}\bm{\psi}(\mathbf{x}) \) with given local point \( \mathbf{x}^i \) in the \( i \)-th iteration of the proposed MM-based algorithm based on the first-order Taylor expansion of $p_1(\mathbf{x})$:
\begin{equation}\small
	\begin{aligned}
	p_1(\mathbf{x})\ge
	\underbrace{\operatorname{Re} \left\{ \bm{\psi}(\mathbf{x}^i)^H \bm{\Sigma} \bm{\psi}(\mathbf{x}) \right\}}_{\triangleq\bar{p}_1(\mathbf{x})} - \underbrace{\bm{\psi}(\mathbf{x}^i)^H \bm{\Sigma} \bm{\psi}(\mathbf{x}^i)}_{\text{constant}}.
	\end{aligned}	\vspace{-0.2cm}
\end{equation}
Apparently, \( {\bar{p}_1(\mathbf{x})} \) is still neither concave nor convex over \( \mathbf{x} \), which cannot be regarded as the surrogate function of \( {{p}_1(\mathbf{x})} \) according to the MM principle \cite{MMorigin}. To address this difficulty, we propose to construct a surrogate function by using the second-order Taylor expansion of \( {{p}_1(\mathbf{x})} \).
Let \( \nabla {\bar{p}_1(\mathbf{x})} \in \mathbb{R}^{N_t} \) and \( \nabla^2 {\bar{p}_1(\mathbf{x})} \in \mathbb{R}^{N_t \times N_t} \) denote the gradient and Hessian matrix of \( {\bar{p}_1(\mathbf{x})} \) w.r.t. \( \mathbf{x} \), respectively. 
Then, we can always find a positive real number \( \delta_1 \) that satisfies \( \delta_1 \mathbf{I}_{N_t} \succeq \nabla^2 {\bar{p}_1(\mathbf{x})} \). 
The detailed expressions of \( \nabla {\bar{p}_1(\mathbf{x})}  \) and \( \nabla^2 {\bar{p}_1(\mathbf{x})}  \) as well as the process to find an appropriate value of $\delta_1$ are given in {{Appendix F}}.
Thus, the following quadratic surrogate function can be employed to globally lower-bound \( \bar{p}_1(\mathbf{x}) \): $\bar{p}_1(\mathbf{x}) \geq \bar{p}_1(\mathbf{x}^i) + \nabla \bar{p}_1(\mathbf{x}^i)^T (\mathbf{x} - \mathbf{x}^i)
- \frac{\delta_1}{2} (\mathbf{x} - \mathbf{x}^i)^T (\mathbf{x} - \mathbf{x}^i)$.
Thus, subproblem (BP1) in the $i$-th iteration of the proposed MM-based algorithm  can be relaxed as\vspace{-0.1cm}
\begin{equation}
	\label{MM1}\small
	\begin{aligned}
		(\text{BP1-m})	\quad \quad
&\max_{\mathbf{x}} \quad {-\frac{\delta_1}{2} \mathbf{x}^T \mathbf{x} + \left(\nabla \bar{p}_1(\mathbf{x}^i) + \delta_1 \mathbf{x}^i\right)^T \mathbf{x}}
\\
&\;\;\;\text{s.t. } \;\;\; \text{(\ref{p6c})},
	\end{aligned}\vspace{-0.1cm}
\end{equation}
where constraint (\ref{p6b}) can be safely omitted because the objective function of subproblem (BP1) already incorporates this constraint and if the final solution of (BP1-m) after a sufficient number of iterations violates (\ref{p6b}), this implies that we need to solve subproblem (BP2) instead.
Since problem (BP1-m) is a typical convex QP problem, its global optimal solution, denoted by $\mathbf{x}^{i+1}$, can be efficiently obtained by off-the-shelf solvers, such as CVX \cite{cvx}.
Therefore, if $\mathbf{x}^{i+1}$ satisfies (\ref{p6b}), we can find a feasible solution of subproblem (SP1), which is also the optimal solution of the original problem (BP). Otherwise, we continue to compute $\mathbf{x}^{i+1}$ in the next iteration. 
The details of the proposed algorithm to solve problem (SP1) are summarized in Algorithm 3, which is guaranteed to converge to the set of stationary solutions \cite{MMconvergence}. 
It is noteworthy that in step 6 of Algorithm 3, if we are unable to obtain a feasible solution for subproblem (SP1), then -1 is returned and in this case, we should turn to subproblem (BP2) for solving the original problem (BP). 
\begin{algorithm}[t]
	\caption{MM for Solving Problem (SP1)}
	\renewcommand\baselinestretch{\linspreadalgr}\selectfont
	\renewcommand{\algorithmicrequire}{\textbf{Input:}}
	\renewcommand{\algorithmicensure}{\textbf{Initialize:}}
	\begin{algorithmic}[1]
		\REQUIRE  $N_t$, $D_x$, $d$, $\theta$, $\theta_t^1$, $\lambda$, $L_t$, $\bm{\sigma}$, $\epsilon_1$.
		\ENSURE $i \leftarrow 0$.

		\REPEAT 
		\STATE Compute $\nabla \bar{p}_1(\mathbf{x}^i)$ and $\nabla^2 \bar{p}_1(\mathbf{x}^i)$. 
		\STATE Update $\delta_1$, 
		and obtain $\mathbf{x}^{i+1}$ by solving (BP1-m).
		\STATE $i \leftarrow i + 1$.
		\UNTIL {$\begin{vmatrix} p_1(\mathbf{x}^{i})-p_1(\mathbf{x}^{i-1})\end{vmatrix}<\epsilon_1$}
		\RETURN -1.
	\end{algorithmic}\vspace{-0.1cm}
\end{algorithm}

\begin{algorithm}[t]
	\caption{Proposed MM-based RGP algorithm}
	\renewcommand\baselinestretch{\linspreadalgr}\selectfont
	\renewcommand{\algorithmicrequire}{\textbf{Input:}}
	\renewcommand{\algorithmicensure}{\textbf{Initialize:}}
	\begin{algorithmic}[1]
		\REQUIRE $N_t$, $D_x$, $d$, $\theta$, $\theta_t^1$, $\lambda$, $L_t$, $\bm{\sigma}$, $\epsilon_2$, \( M \).
		\ENSURE \( i \leftarrow 0 \), obtain \( \mathbf{x}^1 \) via Algorithm 3.
		
		\STATE \textbf{while} \( i < M \) \textbf{do}
		\STATE \;\;\;\; \( i \leftarrow i + 1 \).
		\STATE \;\;\;\;Calculate \( \mathbf{M} \) and \( \mathbf{P}\).
		\STATE \;\;\;\;\textbf{if} \( \| \mathbf{P} \nabla p_2 (\mathbf{x}^i) \|< \epsilon_2 \) \textbf{then}
		\STATE \;\;\;\;\;\;\;\;\textbf{if} \( \mathbf{M} \) is an empty matrix \textbf{then} \textbf{break}.
		\STATE \;\;\;\;\;\;\;\;\textbf{else} \( \mathbf{u} = -(\mathbf{M}\mathbf{M}^T)^{-1} \mathbf{M} \nabla p_2 (\mathbf{x}^i) \), \( u_j \leftarrow \min (\mathbf{u}) \).
		\STATE \;\;\;\;\;\;\;\;\;\;\;\;\textbf{if} \( u_j \geq 0 \) \textbf{then} \textbf{break}.
		\STATE \;\;\;\;\;\;\;\;\;\;\;\;\textbf{else} remove \( \mathbf{M}_j \) and go to step 3. \textbf{end if}
		\STATE \;\;\;\;\;\;\;\;\textbf{end if}
		\STATE \;\;\;\;\textbf{else} obain $\alpha^i$ via Armijo step-size rule.	
		\STATE  \;\;\;\;\;\;\;\;\;\;\;\;\( \mathbf{x}^{i+1}\leftarrow \mathbf{x}^i- \alpha^i\mathbf{P}\nabla p_2 (\mathbf{x}^i)\).
		\STATE \;\;\;\;\textbf{end if}
		\STATE \textbf{end while}
		\RETURN \( \mathbf{x}^{i+1} \).
	\end{algorithmic}\vspace{-0.1cm}
\end{algorithm}

Next, we consider subproblem (BP2).
Since its objective function is a sine function with its domain restricted to $[0,\pi]$ and subject to the constraint (\ref{q15}), i.e., $\upsilon(\mathbf{x})+\phi(\mathbf{x})\ge \frac{\pi}{2}$, it can be equivalently reformulated as \vspace{-0.1cm}
\begin{equation}
	\label{MM2}\small
	\begin{aligned}
		(\text{BP2-m})	\quad \quad\quad
		\min\limits_{\mathbf{x}}\;\;&p_2(\mathbf{x})\triangleq\upsilon(\mathbf{x})+\phi(\mathbf{x})&&\\
		\text{s.t.}\;\;&\text{(\ref{q15})},\text{(\ref{p6c})}.&&
	\end{aligned}\vspace{-0.1cm}
\end{equation}
It is noteworthy that the constraint (\ref{q15}) in problem (BP2-m) is equivalent to \( p_1(\mathbf{x}) \leq \frac{N_t}{P_T} \Gamma \sigma_C^2 \) and it can also be omitted
because if the solution of the resulting problem violates (\ref{q15}), i.e., \( \upsilon(\mathbf{x}) + \phi(\mathbf{x}) < \frac{\pi}{2} \), then this solution becomes a feasible solution of subproblem (SP1), thus serving as the optimal solution to the original problem (BP).
Therefore, by enlarging the feasible region via ignoring (\ref{q15}), we increase the probability of obtaining the optimal solution of the original problem (BP).
Furthermore, by substituting (\ref{q12}) and (\ref{q13}) into $p_2(\mathbf{x})$, $p_2(\mathbf{x})$ can be equivalently rewritten as follows:\vspace{-0.1cm}
\begin{equation}\small
\begin{aligned}
	\label{p_2}
	p_2(\mathbf{x})=
	\arccos\sqrt{\frac{{p_1(\mathbf{x})}}
		{N_t{\|{\mathbf{h}}\|}^2}}+
	\arcsin\sqrt{\frac{\Gamma\sigma_C^2}{P_T\|\mathbf{h}\|^2}},
\end{aligned}\vspace{-0.1cm}
\end{equation}
which is much more complex than \( p_1(\mathbf{x}) \). Thus, finding a proper surrogate function based on the MM principle is quite challenging. 
However, we find that \(\phi(\mathbf{x})\), i.e., $\arcsin\sqrt{{\Gamma\sigma_C^2}/({P_T\|\mathbf{h}\|^2})}$, does not make a significant contribution to the overall optimization process, because \(\|\mathbf{h}\|^2\) is not sensitive to the changes in \(\mathbf{x}\). This can be seen from the following expression of \(\|\mathbf{h}\|^2\):\vspace{-0.1cm}
\begin{equation}
	\label{h2expand}\small
	 \|\mathbf{h}(\mathbf{x})\|^2=
	\sum\limits_{k=1}^{N_t}\begin{vmatrix}
	\sum\limits_{i=1}^{L_t}\sigma_ie^{-j\frac{2\pi}{\lambda}\sin\theta_t^ix_k}
	\end{vmatrix}^2,\vspace{-0.1cm}
\end{equation}
which comprises multiple periodic components with distinct phases and frequencies, resulting in numerous oscillations w.r.t. $ \mathbf{x} $. Despite their relatively small amplitudes, $ \|\mathbf{h}(\mathbf{x})\|^2 $ exhibits frequent but low-magnitude fluctuations, which may induce the optimization algorithm to converge to a bad stationary point.
As shown in the toy example depicted in Fig. \ref{MM_illustration}, directly minimizing \( p_2(\mathbf{x}) \) using a gradient descent type algorithm with a random initial point often converges to a bad stationary point, as indicated by the red circle. In contrast, by minimizing \( \upsilon(\mathbf{x}) \) first and using its solution as the initial point for minimizing \( p_2(\mathbf{x}) \), the performance of the final solution can be significantly improved, as shown by the green circle.
\begin{figure}[!t]
	\centering
	\vspace{-0.4cm}
	\includegraphics[width=2.5in]{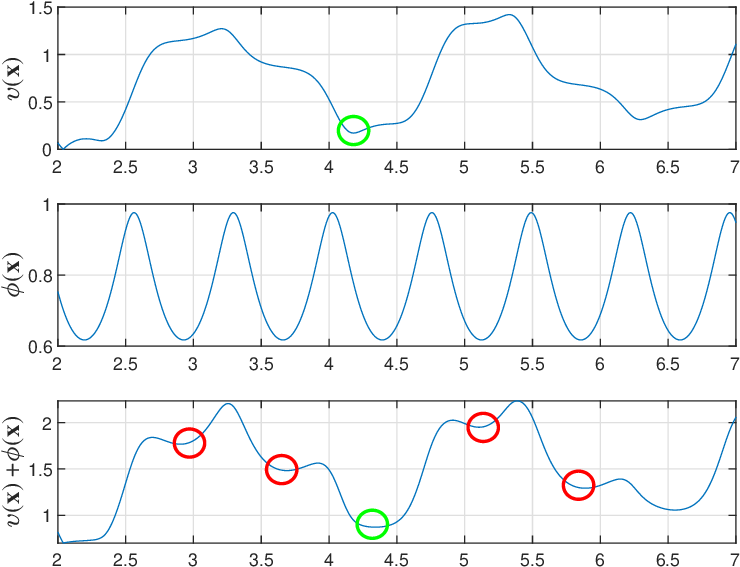}\vspace{-0.2cm}
	\caption{Example of \( \upsilon(\mathbf{x}) \), $\phi(\mathbf{x})$ and \( \upsilon(\mathbf{x})+\phi(\mathbf{x}) \) over ${x_1}$.}
	\label{MM_illustration}
	\vspace{-0cm}
\end{figure}

Consequently, we proceed to analyze \(  \upsilon(\mathbf{x})  \), i.e., $\arccos\sqrt{{{p_1(\mathbf{x})}}/
({N_t{\|{\mathbf{h}}\|}^2})}$.
Note that \( \arccos(x) \) is a monotonically decreasing function of \(x \in [0, 1]\), thus minimizing \( \upsilon(\mathbf{x}) \) is equivalent to maximizing \(\bar{p}_2(\mathbf{x}) \triangleq \frac{{p_1(\mathbf{x})}}{\|{\mathbf{h}}\|^2}\).
Since $\frac{p_1(\mathbf{x})}{\|{\mathbf{h}}\|^2}$ is jointly convex in $\{\bm{\psi},\|{\mathbf{h}}\|^2\}$, we can derive the following lower bound with given local point \( \mathbf{x}^i \) in the \( i \)-th iteration of the proposed MM-based algorithm:\vspace{-0.1cm}
\begin{equation}\small
\begin{aligned}
	\label{MM2RPGM}
	\frac{{\bm{\psi}}(\mathbf{x})^H\bm{\Sigma}\bm{\psi}(\mathbf{x})}{\|{\mathbf{h}(\mathbf{x})}\|^2}\ge
	\frac{2\bar{p}_1(\mathbf{x})}{\|{\mathbf{h}(\mathbf{x}^i)}\|^2}-
	\frac{p_1(\mathbf{x}^i)}{\|{\mathbf{h}(\mathbf{x}^i)}\|^4}\|{\mathbf{h}(\mathbf{x})}\|^2.
\end{aligned}\vspace{-0.1cm}
\end{equation}
Building upon the previous analysis on the insensitivity of ${\|\mathbf{h}(\mathbf{x})\|^2}$ w.r.t. $\mathbf{x}$, we first omit the term $\|\mathbf{h}(\mathbf{x})\|^2$ in (\ref{MM2RPGM}),
where Algorithm 3 can be immediately used to obtain $\mathbf{x}^1$ as an initial point for problem (BP2-m). Subsequently, an MM-based RGP algorithm (summarized in Algorithm 4) is proposed to solve problem (BP2-m), where a projection matrix 
\(
\mathbf{P} = \mathbf{I} - \mathbf{M}^T (\mathbf{M}\mathbf{M}^T)^{-1} \mathbf{M}
\) 
(\( \mathbf{M} \) is composed of the row vectors from \( \mathbf{U} \) that satisfy \(\mathbf{U}_k \mathbf{x}^i = (\mathbf{l}_u)_k \) in the $i$-th iteration) is introduced to project the updated gradient descent step back onto the feasible region. Furthermore, we denote the gradient vector of \( {{p}_2(\mathbf{x})} \) w.r.t. \( \mathbf{x} \) by \( \nabla {{p}_2(\mathbf{x})} \in \mathbb{R}^{N_t} \), which can be simply obtained by the chain rule.
With the proposed initialization method, Algorithm 4 only requires a relatively small number of iterations to converge to a good stationary point of problem (BP2-m) \cite{rosen}.

\begin{remark}
In the most general scenario where both the transmit and receive antennas are movable, it can be readily seen that the transmit and receive MAs can be designed separately without loss of optimality and the corresponding problems can be efficiently solved by using the techniques proposed in the above two special cases, respectively. 
{Moreover, practical MA platforms admit only a finite grid of positions due to limited rail length and actuator resolution. Therefore, one can first compute the continuous-domain antenna positions via our proposed algorithm and then round each position to its nearest grid point, achieving near-optimal performance with low complexity.}
\end{remark}

\vspace{-0.2cm}
\subsection{CRB Performance Analysis}
In this subsection, we analyze the impacts of movable transmit antennas on the sensing performance under varying user communication SNR thresholds.
We can see that the constraints of subproblems (SP1) and (SP2), i.e., (\ref{p6b}) and (\ref{p7b}), are both related to the user's communication SNR threshold, $\Gamma$. 

Consequently, we can regard subproblem (SP1) as a sensing performance optimization problem under low communication SNR threshold. 
By substituting the optimal beamforming solution \(\mathbf{w}_\text{opt} = \sqrt{P_T}\frac{\mathbf{a}}{\|\mathbf{a}\|}\) into the CRB expression, i.e., (\ref{CRB_expression}), we can obtain the global minimum of the CRB in terms of \( \mathbf{w} \) and \( \mathbf{x} \) as\vspace{-0.1cm}
\begin{equation}
	\label{CRB_expression_min}\small
	\begin{aligned}
		\text{CRB}(\theta)_\text{min}=
		\frac{\sigma_R^2/
			(2\begin{vmatrix}
				\alpha
			\end{vmatrix}^2L{N_tP_T})}
		{
			(\frac{2\pi}{\lambda}\cos\theta)^2
			\big(\sum\limits_{i=1}^{N_r}y_i^2
			-\frac{1}{N_r}(\sum\limits_{i=1}^{N_r}y_i)^2\big)
		}.
	\end{aligned}\vspace{-0.1cm}
\end{equation}
However, due to the transmit power budget constraint on the beamforming vector, i.e., $\|\mathbf{w}\|^2\le{P_T}$, this minimum value cannot be sustained as the user's communication SNR threshold $\Gamma$ increases.
Specifically, as $\Gamma$ increases, the constraint set of (\ref{p5b}) in problem (BP) shrinks, which will prevent the beamforming vector from remaining at $\sqrt{P_T}\frac{\mathbf{a}}{\|\mathbf{a}\|}$ and finally leads to sensing performance degradation. 
In this case, by deploying movable transmit antenna array, an additional degree of freedom is introduced by allowing the CRB to remain at its lowest level (\ref{CRB_expression_min}) until the user's SNR threshold reaches \(\Gamma_0=\frac{|{{\mathbf{h}(\mathbf{x}^{\star})^H}{\mathbf{a}(\mathbf{x}^{\star})}}|^2}{N_t\sigma_C^2/P_T} \) (as indicated by either constraint (\ref{p6b}) or (\ref{p7b})), where $\mathbf{x}^{\star}$ is the optimal transmit APV obtained by Algorithm 1 or Algorithm 3.
To be specific, the movable transmit antenna array can increase $\Gamma$ by \(\Delta_{\Gamma}=20\lg(\frac{|{{\mathbf{h}(\mathbf{x}^{\star})^H}{\mathbf{a}(\mathbf{x}^{\star})}}|}
{|{{\mathbf{h}^H}{\mathbf{a}}}|_\text{ULA}} )\) dB while maintaining the same sensing performance.

In contrast, subproblem (SP2) addresses the high communication SNR scenario. If no feasible solution exists for subproblem (SP1), i.e., $\Gamma > \Gamma_0$, the solution of subproblem (SP2) results in a higher CRB than the minimum value specified in (\ref{CRB_expression_min}). Therefore, we can intuitively conclude that when the user's communication SNR threshold exceeds $\Gamma_0$, the sensing performance will inevitably degrade.

The above analysis embodies the trade-off between sensing and communication performance inherent in ISAC systems, and it shows that movable transmit antenna arrays can significantly enhance both performance by introducing an additional degree of freedom.
\vspace{-0.2cm}
\subsection{Computational Complexity Analysis}
In this subsection, we analyze the computational complexity of the proposed Algorithms 1-4, where we focus exclusively on the number of multiplication operations.

First, we consider the worst-case complexity of  Algorithms 1 and 2 (i.e., the BT-BFS and BT-DFS algorithms).
{Since the complexity of calculating the value of $D_\text{min}$ in (\ref{Dmin}) or (\ref{Dmin_2}) is \(\mathcal{O}(1)\), the worst-case complexity of the BT-BFS algorithm is} \vspace{-0.2cm}
{{
		\begin{equation}\small
		\begin{aligned}
			\mathcal{C}_{\rm BT\text{-}BFS}&=\mathcal{O}\Bigl(\sum_{c=1}^{N_t-1} \binom{N_t}{c}\Bigr)\\
			&=
			\mathcal{O}\Bigl(\sum_{c=0}^{N_t} \binom{N_t}{c}
			-\binom{N_t}{0}-\binom{N_t}{N_t}\Bigr)\\
			&=
			\mathcal{O}\Bigl(2^{N_t}-2\Bigr)=
			\mathcal{O}\Bigl(2^{N_t}\Bigr).
		\end{aligned}
	\end{equation}
Second, the BT-DFS algorithm activates constraints one by one along a single depth-first path and terminates as soon as \(D_x\ge D_{\min}\).  In the worst case, it examines at most \(N_t-1\) constraints, each incurring \(\mathcal{O}(1)\) complexity for \(D_{\min}\) computation and closed-form update (\ref{pcase1opt}), hence  
$
\mathcal{C}_{\rm BT\text{-}DFS}
=\mathcal{O}\bigl(N_t\bigr).
$}}

Then, for Algorithm 3, its computational complexity is dominated by solving the QP problem (BP1-m) which is $\mathcal{O}(N_t^{1.5} \ln(1/\beta))$ with $\beta$ denoting the accuracy of the interior-point method \cite{book3}. Let $\gamma_r$ denote the maximum iteration number, the total complexity of Algorithm 3 can be expressed as $\mathcal{O}(N_t^{1.5} \ln(1/\beta)\gamma_r)$. 
For Algorithm 4, the complexity of calculating $\nabla_{\mathbf{x}}{p_2{(\mathbf{x})}}$ is $\mathcal{O}(N_tL_t^2+2L_t^2+3N_tL_t)$, thus the overall complexity of Algorithm 4 is $\mathcal{O}(M(N_tL_t^2+2L_t^2+3N_tL_t))$.

\vspace{-0.2cm}
\section{Numerical Results}
This section presents numerical results to evaluate the effectiveness of the proposed algorithms for optimizing the sensing performance of the MA-enabled ISAC system while ensuring the communication quality. 
Without loss of generality, we consider a MIMO ISAC BS that is equipped with \( N_t = 18 \) and \( N_r = 20 \) MAs at its transmitter and receiver, respectively. The power budget is \( P_T = 20 \) dBm, the noise power is set as \( \sigma_C^2 = \sigma_R^2 = 0 \) dBm, and the frame length is set as \( L = 30 \). The apertures of both transmit and receive antenna arrays are set to be \( D_x = D_y = 13.55 \lambda \) unless otherwise specified. \footnote{{We chose such values of $D_x$ and $D_y$ to place BT-BFS in its most computationally demanding regime, forcing it to explore the maximum number of active set combinations before feasibility is achieved. Although larger apertures improve sensing performance by relaxing constraints, they trivialize the optimization by immediately satisfying feasibility via the closed-form solution in (\ref{q17}). Therefore, such apertures prevent premature termination and enable fair validation of the algorithm’s effectiveness and robustness.}}
The geometry channel model is considered \cite{10243545}, where the number of transmit paths is set to $L_t = 18$. The path response vector is assumed to follow Rician fading with ${\sigma}_1 \sim \mathcal{CN}(0, \kappa / (\kappa + 1))$ and $\sigma_p \sim \mathcal{CN}(0, 1 / ((\kappa + 1)(L_t - 1)))$ for $p = 2, 3, \dots, L_t$, where $\kappa = 3$ denotes the ratio of the average power for LoS paths to that for NLoS paths. The azimuth AoDs are assumed to be i.i.d. variables that are uniformly distributed over $[-\pi/2, \pi/2]$. The minimum distance between the MAs is set as $d = \lambda / 2$.
Other parameters are set as follows unless otherwise specified: $\epsilon_1 = \epsilon_2 = 10^{-3}$ and $\theta = 0 ^\circ$.
The following benchmark schemes are considered in our simulations: 1) ULA with half-wavelength spacing (ULAH); 2) ULA with full aperture (ULAF), i.e., the inner-antenna spacings for the transmit and receive antennas are \( D_x / (N_t-1) \) and \( D_y / (N_r-1) \), respectively;
3) a scheme based on the successive convex approximation (SCA) technique for solving problems (AP-m), (BP1) and (BP2);
4) a scheme based on RGP algorithm with random initialization for solving problems (BP1) and (BP2).

\begin{figure}[t]
	\centering
	\includegraphics[width=3in]{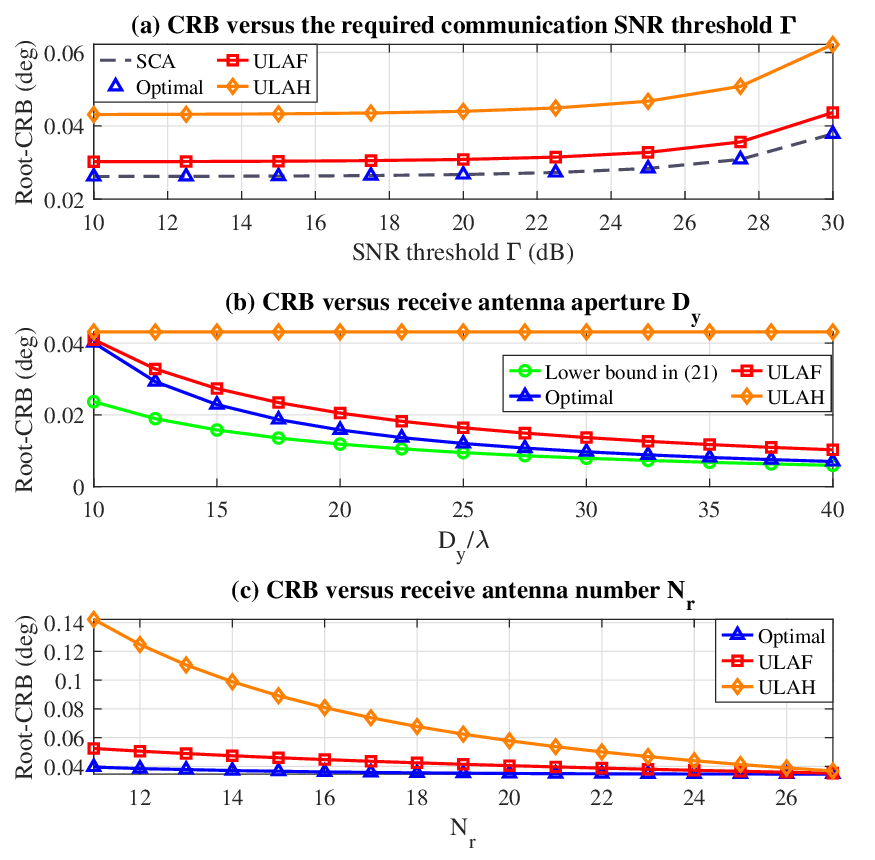}\vspace{-0.4cm}
	\caption{CRB versus the required communication SNR threshold $\Gamma$, antenna aperture $D_y$ and antenna number $N_r$ in the case of receive MAs.}
	\label{Rx_CRB_SNR}
	\vspace{-0cm}
\end{figure}

\begin{figure}[t]
	\centering
	\includegraphics[width=3.2in]{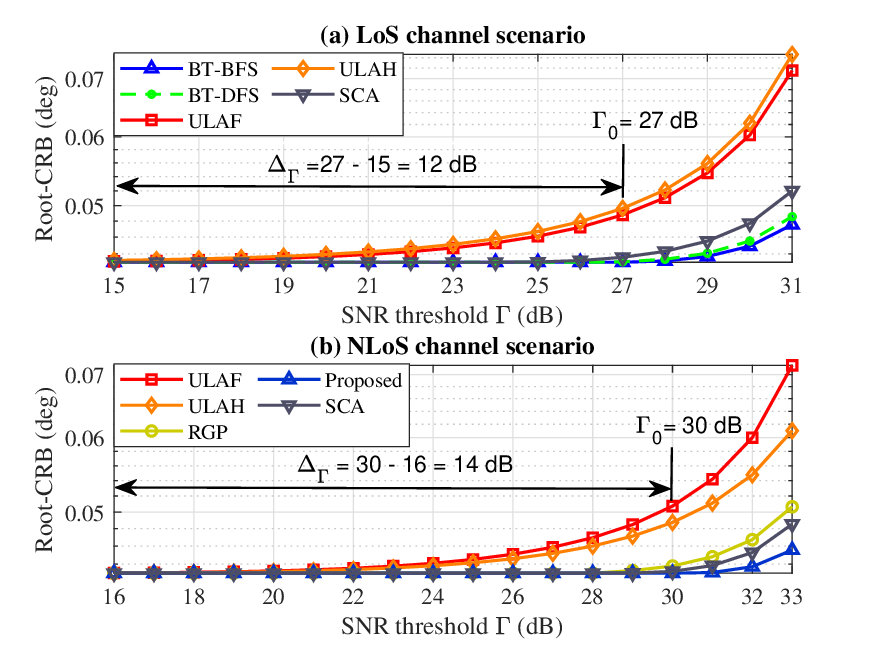}\vspace{-0.4cm}
	\caption{CRB versus the required communication SNR threshold $\Gamma$ in the case of transmit MAs.}
	\label{LOS_CRB_SNR}
	\vspace{-0cm}
\end{figure}

\begin{figure}[t]
	\centering
	\includegraphics[width=3.5in]{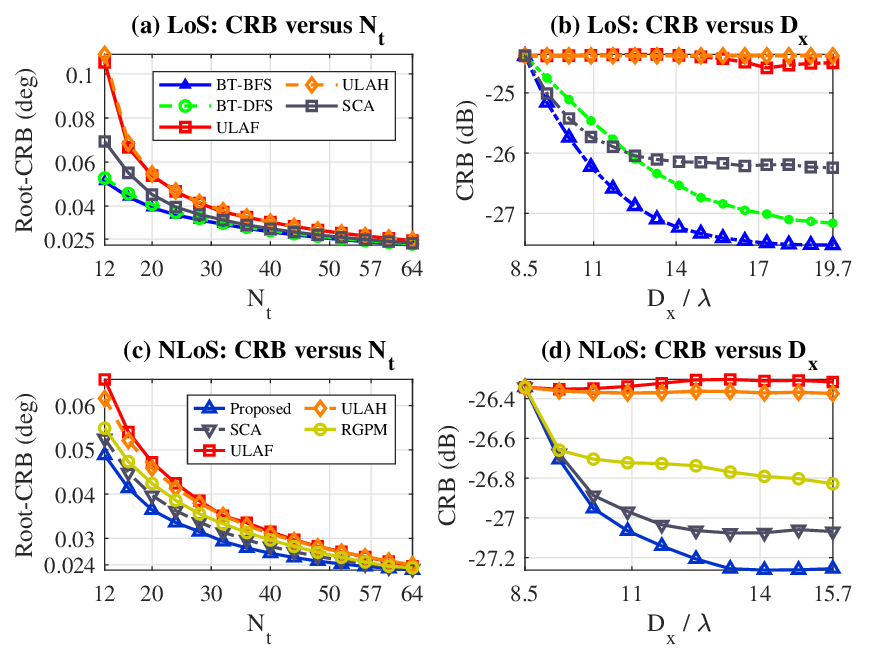}\vspace{-0.4cm}
	\caption{{{CRB versus antenna number $N_t$ and antenna aperture $D_x$ in the case of transmit MAs.}}}
	\label{LOS&NLOS_CRB_Nt_Dx}
	\vspace{-0cm}
\end{figure}

%
%
%

\begin{figure}[t]
	\centering
	\includegraphics[width=3in]{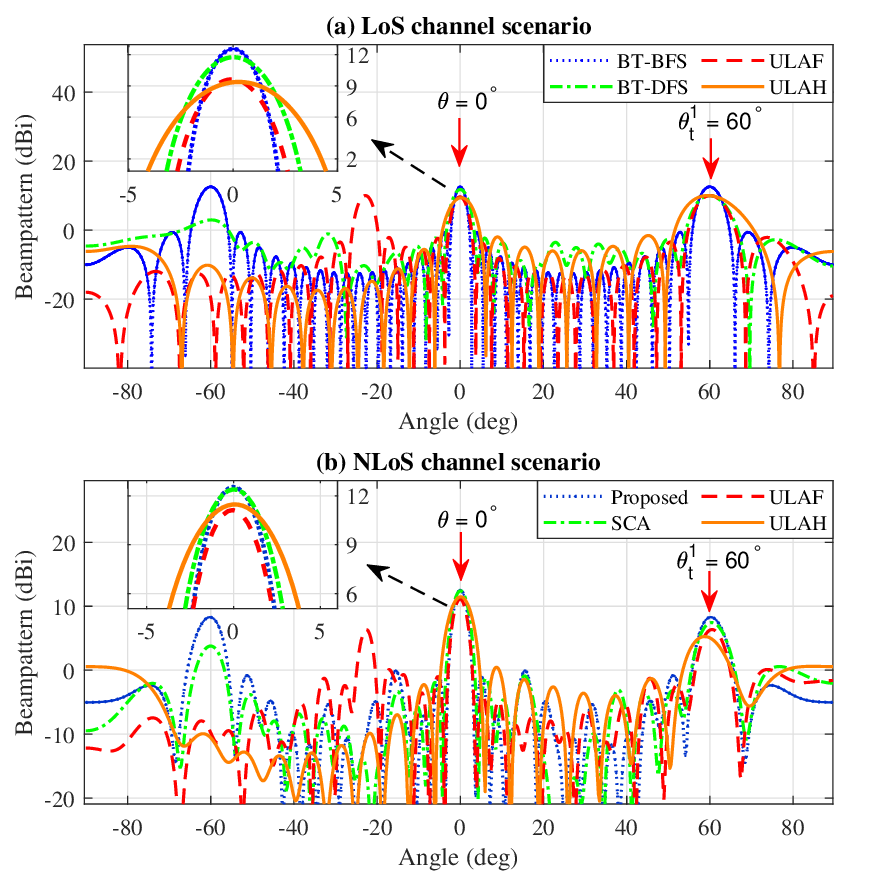}\vspace{-0.4cm}
	\caption{Comparison of beampatterns with different antennas' positions.}
	\label{LOS_beampattern}
	\vspace{-0cm}
\end{figure}

%
%

\begin{figure}[t]
	\centering 	\vspace{-0.4cm}
	\includegraphics[width=3.5in]{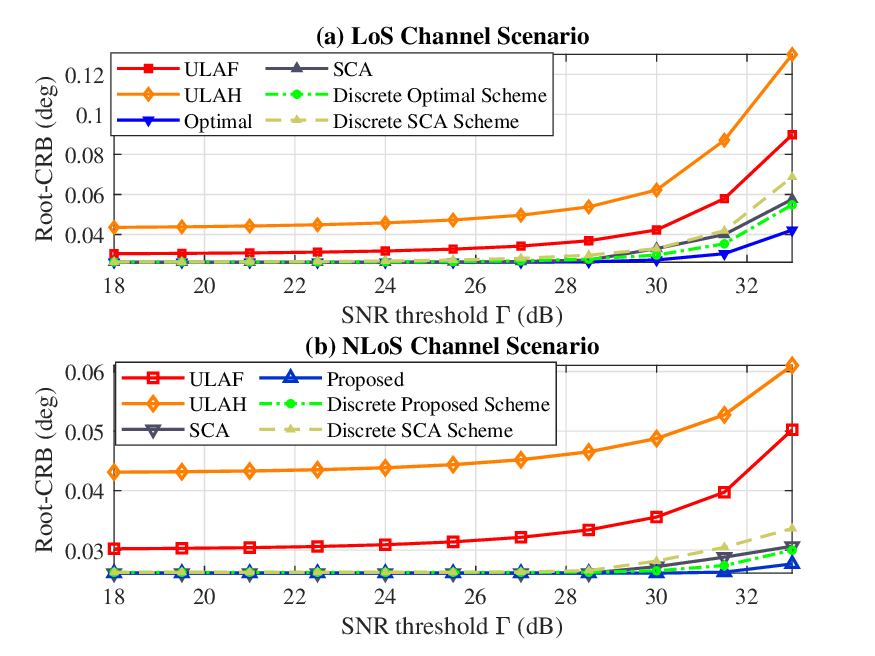}\vspace{-0.4cm}
	\caption{{{CRB versus the required communication SNR threshold $\Gamma$ in the case of both transmit and receive MAs.}}}
	\label{CRB_SNR_TxRx}
	\vspace{-0.1cm}
\end{figure}

%
%
%
%
%


\vspace{-0.2cm}
\subsection{Receive MA Case}
Fig. \ref{Rx_CRB_SNR} demonstrates the sensing performance (in terms of root-CRB) achieved by movable receive antennas under various user communication SNR thresholds, receive antenna apertures and numbers.
As shown in Fig. \ref{Rx_CRB_SNR} (a), the root-CRB performance achieved by the optimal receive MA positions given in (\ref{q10}) is significantly lower than those of ULAH and ULAF with the same antenna aperture.
In the meantime, from Fig. \ref{Rx_CRB_SNR} (b), we can see that increasing the antenna aperture from \(D_x = 10\lambda\) to \(D_x = 40\lambda\) is able to reduce the root-CRB consistently, which validates the analysis in Section \uppercase\expandafter{\romannumeral3}.
Furthermore, as $D_x$ increases, the performance of the proposed scheme will gradually approach the lower-bound root-CRB given in (\ref{approx1}), and a 4.77 dB CRB gain can be achieved over the ULAF scheme.
In Fig. \ref{Rx_CRB_SNR} (c), it is observed that under a given antenna aperture $ D_y = 13.55\lambda $, increasing the number of receive antennas improves the sensing performance, but the performance gain cannot grow unboundedly. This is because the value of $ N_t$ is limited by $ D_y / d + 1$ with fixed $D_y$, corresponding to the case that all antennas are arranged in a ULA with a spacing of $ \lambda/2 $.
\vspace{-0.4cm}
\subsection{Transmit MA Case}
Next, we evaluate the effectiveness of our proposed algorithms for movable transmit antennas, as illustrated in Fig. \ref{LOS_CRB_SNR}.
From Fig. \ref{LOS_CRB_SNR} (a), we can see that the proposed BT-BFS algorithm achieves the minimum CRB among the considered schemes across all user communication SNR thresholds. 
Besides, the proposed BT-BFS algorithm is able to maintain the minimum CRB level of target angle estimation specified in (\ref{CRB_expression_min}) until the SNR threshold reached about $\Gamma_0 = 27$ dB, which is $\Delta_{\Gamma}=12$ dB larger than those of the traditional ULAs (about 15 dB). Additionally, the proposed low-complexity BT-DFS algorithm also significantly outperforms the traditional ULAs as well as the SCA scheme. 
Under the NLoS channel scenario (Fig. \ref{LOS_CRB_SNR} (b)), the proposed MM-based RGP algorithm maintains the CRB at the minimum specified in (\ref{CRB_expression_min}) across an SNR range of \(\Delta_{\Gamma} = 14\) dB, which is also larger than those of the benchmark schemes.
Furthermore, we can see that initializing the RGP by the solution derived from the MM algorithm (as in Algorithm 3) yields superior performance as compared to random initialization, which validates the superiority of our proposed initialization method.

{Next, we show in Fig. \ref{LOS&NLOS_CRB_Nt_Dx} the sensing performance achieved by the proposed algorithms under various transmit antenna numbers and apertures. 
From Fig. \ref{LOS&NLOS_CRB_Nt_Dx} (a) and Fig. \ref{LOS&NLOS_CRB_Nt_Dx} (c), we can observe that increasing the number of transmit antennas, i.e., \(N_t\), enhances the sensing performance (measured by root-CRB) of all schemes, when the antenna number is extended up to 64 and the aperture is set as $64\lambda$.
An interesting phenomenon emerges as $N_t$ grows large: the performance gap between all schemes diminishes, and their CRBs become similar. This occurs because the communication SNR threshold $\Gamma = 30$ dB becomes easier to satisfy with larger antenna arrays. 
This trend is consistent with the analysis in Section IV. D, where increasing $N_t$ enables all schemes to achieve the optimal beamforming vector $\mathbf{w}_{\text{opt}} = \sqrt{P_T} \mathbf{a} / \|\mathbf{a}\|$, thereby achieving the minimum CRB, i.e., (\ref{CRB_expression_min}),  and decreasing at an order of $O(1/N_t)$.
Moreover, it is important to mention that this result does not  imply a diminished advantage of the MA scheme; as discussed in Section IV. D and illustrated in Fig. \ref{LOS_CRB_SNR}, the MA scheme’s benefits become more pronounced at higher SNR levels, particularly given that an antenna number of 64 is already quite large.
Fig. \ref{LOS&NLOS_CRB_Nt_Dx} (b) and Fig. \ref{LOS&NLOS_CRB_Nt_Dx} (d) show that increasing the transmit antenna aperture \(D_x\) only improves the sensing performance of the MA-enabled scheme, and the root-CRB performance eventually saturates with the increasing of $D_x$, which is consistent with our analysis in Section \uppercase\expandafter{\romannumeral4}.  Moreover, we can see that our proposed algorithms outperform all benchmark schemes under different settings.}

Then, to gain more insights, we illustrate in Fig. \ref{LOS_beampattern} the beampatterns of different schemes where the required SNR level is set to be $\Gamma = 30$ dB and the azimuth AoD of the LoS path is set to $\theta_t^1 = {\pi}/{3}$.
Under the LoS channel scenario, as illustrated in Fig. \ref{LOS_beampattern} (a), all four beamformers direct their mainlobes towards $0^\circ$ (the target direction), and allocate a large amount of  energy at $60^\circ$ (the communication user's direction). 
The proposed BT-BFS algorithm achieves the highest peak power and the narrowest mainlobe width towards the target angle, followed by the BT-DFS algorithm. 
Both proposed algorithms demonstrate substantial improvements over the traditional ULAs.  
Under the NLoS scenario, we observe that although the considered four schemes still direct their mainlobes towards $0^\circ$, the energy allocated to $60^\circ$ is reduced as compared to the LoS case due to dispersion across other NLoS paths (Fig. \ref{LOS_beampattern} (b)). The proposed MM-based RGP algorithm exhibits the highest peak power and the narrowest mainlobe width at the target direction, followed by the SCA algorithm. Similar to the LoS case, both algorithms significantly outperform the traditional ULAs. 

{{
		\vspace{-0.3cm}
\subsection{General Case}
Finally, Fig. \ref{CRB_SNR_TxRx} illustrates the relationship between root-CRB performance and user communication SNR threshold when both transmit and receive antennas are movable.
Under the LoS channel scenario, the optimal scheme employs the BT-BFS algorithm for transmit movable antenna position optimization, while under the NLoS channel scenario, the proposed scheme utilizes the MM-based RGPM algorithm.
Under both scenarios, the optimal solution in (\ref{q10}) is utilized for receive antenna position optimization.  
As can be observed from Fig. \ref{CRB_SNR_TxRx}, MA arrays significantly enhance the target CRB performance compared to the considered benchmarks at the same user communication SNR threshold.
In both LoS and NLoS scenarios, we also include two discrete schemes obtained by quantizing the optimal (proposed) scheme and the SCA scheme. These discrete schemes comply with the antenna spacing constraints specified in (\ref{q4}) and (\ref{q5}), using a quantization step size of $\lambda/5$.
As shown in Fig. \ref{CRB_SNR_TxRx}, both the discrete optimal and discrete proposed schemes achieve better performance compared to the ULA and SCA schemes, which demonstrates their practical effectiveness.
}}


\vspace{-0.3cm}
\section{Conclusion}\vspace{-0.1cm}
This paper presented an MA-enabled ISAC system for future wireless networks, and provided optimal solutions and efficient algorithms for joint antenna position and beamforming optimization.
Specifically, for the case with receive MAs, we derived an optimal antenna position solution, which is able to achieve a CRB gain of at most 4.77 dB over the traditional FPAs. 
While for the case with transmit MAs, we proposed the BT-BFS and BT-DFS algorithms to obtain global optimal and local optimal solutions in the LoS channel scenario, respectively. 
In the NLoS scenario, we also proposed an MM-based RGP algorithm (together with an efficient initialization method) to obtain stationary solutions of the considered problem. 
Numerical results demonstrate the superiority of transmit/receive MAs in improving the ISAC performance over traditional FPAs, and the effectiveness of the proposed algorithms.
{Although our framework guarantees global optimality in linear array, single‐user, single‐target setups, extending it to 2D/3D array and multi‐user, multi‐target scenarios introduces joint optimization challenges that typically lead only to stationary solutions. 
We have already explored one such extension of a multi‐user downlink ISAC system in our prior work \cite{Han2025VTC}; accordingly, exploring scalable approaches to address these challenges in even more general settings is a promising direction for future research.
}

\vspace{-0.2cm}
{\appendices
	\vspace{-0.1cm}
\section{Proof of Lemma 1}\vspace{-0.1cm}
To prove this lemma, we first simplify $f(\mathbf{y})$ and express it in a quadratic form given by $f(\mathbf{y})=\mathbf{y}^T\mathbf{Q}\mathbf{y}$,
where $\mathbf{Q}=\mathbf{I}_{N_r}-\frac{1}{N_r}\mathbf{J}_{N_r}$.
Then, for any $a\in\mathbb{R}$, we have
$		f(\mathbf{y}+a)
=\mathbf{y}^T\mathbf{Q}\mathbf{y}+a\mathbf{y}^T\mathbf{Q}\mathbf{1}
+a\mathbf{1}^T\mathbf{Q}\mathbf{y}+a^2\mathbf{1}^T\mathbf{Q}\mathbf{1}
{=}f(\mathbf{y}),$
where the second equality holds because $\mathbf{1}^T\mathbf{Q}=\mathbf{0}^T$ and $\mathbf{1}^T\mathbf{Q}\mathbf{1}=0$. This completes the proof.

\vspace{-0.3cm}
\section{Proof of Lemma 3}\vspace{-0.1cm}
First, we analyze the property of the local maxima of $f(\mathbf{y})$ by examining its second-order derivatives, which can be obtained as
$	\frac{\partial^2{f(\mathbf{y})}}{\partial{y_i}^2}=2-\frac{2}{N_r},\;
\frac{\partial^2{f(\mathbf{y})}}{\partial{y_i}\partial{y_j}}=-\frac{2}{N_r}(i\neq{j})$.
It can be readily seen that the Hessian matrix of $f(\mathbf{y})$ 
is a positive definite matrix, and thus 
$f(\mathbf{y})$ is a convex continuous function with its maximum value always achieved at the boundary points. 
This implies that when 
$\mathbf{y} \backslash y_i $
are determined, the maximum value of $f(\mathbf{y})$ must be attained at the boundary of the feasible region of \( y_i \), i.e., $y_i=y_{i-1}+d\;\;\text{or}\;\;y_i=y_{i+1}-d.$

Next, we prove that compared to $y_i=y_{i+1}-d$, $y_i=y_{i-1}+d$ is better in terms of maximizing \( f(\mathbf{y}) \) when \( i\le\left\lfloor\frac{{N_r}}{2}\right\rfloor \).
Due to the translational invariance property of \( f(\mathbf{y}) \) as proved in Lemma 1, we can assume \( y_1 = 0 \) without loss of generality. Now, let
$\hat{y}_k=y_{k-1}+d$ and $\tilde{y}_k=y_{k+1}-d$, \( 2\le{k}\le\left\lfloor\frac{{N_r}}{2}\right\rfloor \), which satisfy $\hat{y}_k\neq{\tilde{y}_k}$.\footnote{If $ \hat{y}_k = \tilde{y}_k $, then $ y_{k+1} = y_{k-1} + 2d $ holds, which implies that $y_k$ is a constant.} 
By substituting  $\hat{\mathbf{y}} \triangleq [y_1,\cdots,y_{k-1},\hat{y}_k,y_{k+1},\cdots,y_{N_r}]^T$ and $\tilde{\mathbf{y}} \triangleq[y_1,\cdots,y_{k-1},\tilde{y}_k,y_{k+1},\cdots,y_{N_r}]^T$ into $f(\mathbf{y})$, we obtain
$		f(\hat{\mathbf{y}})=
(\sum\nolimits_{\substack{{i=1},i\neq{k}}}^{N_r}{y_i^2}+\hat{y}_k^2)
-\frac{1}{N_r}(\sum\nolimits_{\substack{{i=1},  i\neq{k}}}^{N_r}{y_i}+\hat{y}_k)^2$ and 
$f(\tilde{\mathbf{y}})=
(\sum\nolimits_{\substack{{i=1},i\neq{k}}}^{N_r}{y_i^2}+\tilde{y}_k^2)
-\frac{1}{N_r}(\sum\nolimits_{\substack{{i=1},i\neq{k}}}^{N_r}{y_i}+\tilde{y}_k)^2.$
Hence, to establish Lemma 3,  we need to prove that $f(\hat{\mathbf{y}})>f(\tilde{\mathbf{y}})$ holds for  \( 2\le{k}\le\left\lfloor\frac{{N_r}}{2}\right\rfloor \).
By subtracting $f(\hat{\mathbf{y}})$ from $f(\tilde{\mathbf{y}})$, we can obtain
$f(\tilde{\mathbf{y}})-f(\hat{\mathbf{y}})=
\frac{y_{k-1}-{y}_{k+1}+2d}{N_r}
[(N_r-1)(y_{k-1}+{y}_{k+1})-2\sum\nolimits_{\substack{{i=1},i\neq{k}}}^{N_r}y_i].$
Since ${y}_{k+1}\ge{y_{k}+d}$, ${y}_{k}\ge{y_{k-1}+d}$ and $\hat{y}_k\neq{\tilde{y}_k}$, we have $y_{k-1}-{y}_{k+1}+2d<0$. 
Therefore, proving $f(\hat{\mathbf{y}})-f(\tilde{\mathbf{y}})>0$ is equivalent to proving $\left(N_r-1\right)\left(y_{k-1}+{y}_{k+1}\right)-2\sum\nolimits_{i=1,i\neq{k}}^{N_r}y_i<0$.
To prove the latter, we resort to the following result:\vspace{-0.2cm}
\begin{equation}
	\label{L3}\small
	\begin{aligned}
		\sum\limits_{{{i=1},i\neq{k}}}^{N_r}y_i
		\ge\sum\limits_{i=1}^{k-1}(i-1)d+\sum\limits_{i=k+1}^{N_r}(i-1)d
		\overset{(b)}{>}\frac{(N_r-1)^2}{2}d,
	\end{aligned}\vspace{-0.2cm}
\end{equation}
where the inequality $(b)$ holds due to the fact that $k\le\lfloor\frac{N_r}{2}\rfloor\le\frac{N_r}{2}<\frac{N_r+1}{2}$.
Since $y_{k-1}+y_{k+1}\le(k-2)d+kd<(N_r-1)d$ holds, by combining it with (\ref{L3}), we can obtain $\left(N_r-1\right)\left(y_{k-1}+{y}_{k+1}\right)-2\sum\nolimits_{i=1,i\neq{k}}^{N_r}y_i
<0$,
which directly leads to $f(\hat{\mathbf{y}})>f(\tilde{\mathbf{y}})$. This thus completes the proof.


\vspace{-0.3cm}
\section{Proof of Lemma 4}\vspace{-0.1cm}
First, we show that \( \bm{\beta^*} \) is a local maximum of \( E(\bm{\beta}) \) if and only if \( \bm{\beta^*} \) is a local maximum of \( E^2(\bm{\beta}) \).
Given that \( E(\bm{\beta}) \ge 0 \), if \( \bm{\beta^*} \) is a local maximum of \( E(\bm{\beta}) \), then there must exist a neighborhood \( \mathcal{N} \) of \( \bm{\beta^*} \) such that \( E(\bm{\beta^*}) \ge E(\bm{\beta}) \) for \( \forall\bm{\beta} \in \mathcal{N} \), which is equivalent to \( E^2(\bm{\beta^*}) \ge E^2(\bm{\beta}) \) and thus \( \bm{\beta^*} \) is a local maximum of \( E^2(\bm{\beta}) \), and vice versa. Consequently, we only need to verify that \( ( \sum\nolimits_{i=1}^{N} \rho_i )^2 \) is the unique local maximum value of \( E^2(\bm{\beta}) \).

Then, we prove the existence of the local maximum value $(\sum\nolimits_{i=1}^{N}\rho_i)^2$ of \( E^2(\bm{\beta}) \) and that the corresponding solution satisfies condition (\ref{Econ}), where we only need to verify that \(\nabla E^2(\bm{\beta}) =0\) and \(\nabla^2 E^2(\bm{\beta}) \prec 0\) when condition (\ref{Econ}) is satisfied.
Since \( E^2(\bm{\beta}) \) is twice continuously differentiable, we can simply verify that the each entry of the gradient of $E^2(\bm{\beta})$, i.e., \vspace{-0.2cm}
\begin{equation}
	\label{Egradient}\small
	\begin{aligned}
		\frac{\partial}{\partial{\beta_i}}{E^2(\bm{\beta})}=-2\sum\limits_{k=1}^{N}
		\rho_k\rho_i\sin({{\beta}_i-{\beta}_k}),
	\end{aligned}\vspace{-0.1cm}
\end{equation}
is zero when (\ref{Econ}) is satisfied. 
The second-order derivatives of $E^2(\bm{\beta})$ can be similarly determined as 
$		\frac{\partial^2}{\partial{\beta_i}^2}{E^2(\bm{\beta})}=-2\sum\nolimits_{k=1,k\neq{i}}^{N}\rho_k\rho_i\cos({{\beta}_i-{\beta}_k})$ and $
\frac{\partial^2}{\partial{\beta_i}\partial{\beta_j}}{E^2(\bm{\beta})}=2\rho_j\rho_i\cos({{\beta}_i-{\beta}_j}),(i\neq{j})$.
We can also obtain that when (\ref{Econ}) is satisfied, the following inequality holds:
$\mathbf{x}^T \nabla^2 \left(E^2(\bm{\beta})\right) \mathbf{x} = -2 \sum_{i\neq{j}}  \rho_i \rho_j \cos(\beta_i - \beta_j) (x_i - x_j)^2 < 0,$ $\forall\mathbf{x} \in \mathbb{R}^N$
if $\mathbf{x} \neq \mathbf{0} $.
This implies that the Hessian matrix \(\nabla^2 \left(E^2(\bm{\beta})\right)\) is negative definite. 
Thus, by substituting (\ref{Econ}) into $E^2(\bm{\beta})$, the local maximum value $(\sum\nolimits_{i=1}^{N}\rho_i)^2$ can be obtained.

Next, we prove the uniqueness of this local maximum value.
From the triangle inequality, we have 
$|\sum\nolimits_{i=1}^{N}\rho_ie^{j\beta_i}|^2\le
( {\sum\nolimits_{i=1}^{N}|\rho_ie^{j\beta_i}|})^2 =( \sum\nolimits_{i=1}^{N}\rho_i)^2$.
Therefore, this local maximum is also the global maximum of \( E^2(\bm{\beta}) \). Now, we prove that there is no other local maximum that is smaller than \( (\sum\nolimits_{i=1}^{N}\rho_i)^2 \) by contradiction.
Suppose that there exists another local maximum solution $\bm{\beta}^0$, and the corresponding local maximum value $E^2(\bm{\beta}^0)$ satisfies $E^2(\bm{\beta}^0)< (\sum\nolimits_{i=1}^{N}\rho_i)^2$.
Since $\bm{\beta}^0$ does not satisfy condition (\ref{Econ}), it must belong to one of the following two cases.

\textbf{Case}\;\textbf{\uppercase\expandafter{\romannumeral1}}:
For all \( N \) complex numbers, i.e., \( \rho_i e^{j\beta^0_i} \), \( 1 \leq i \leq N \), they are not all collinear, which indicates that there exists a positive integer \(1\le{t}\le{N-1}\) such that \(0<{\eta^0}<{\pi}\), where $\eta^0=\min\left\lbrace \varphi,2\pi-\varphi \right\rbrace$ and 
$\varphi= |\angle(\sum\nolimits_{i=1}^{t}\rho_ie^{j\beta_i^0})
	-\angle(\sum\nolimits_{i=t+1}^{N}\rho_ie^{j\beta_i^0})|$.
Additionally, using the cosine rule, we can obtain
$E^2(\bm{\beta}^0)={\begin{vmatrix}E_1\end{vmatrix}^2+
	\begin{vmatrix}E_2\end{vmatrix}^2+
	2\begin{vmatrix}E_1\end{vmatrix}\begin{vmatrix}E_2\end{vmatrix}\cos{\eta^0}}$,
where $E_1\triangleq\sum\nolimits_{i=1}^{t}\rho_ie^{j\beta_i^0}$ and $E_2\triangleq\sum\nolimits_{i=t+1}^{N}\rho_ie^{j\beta_i^0}$.
Since \(\bm{\beta}^0\) is a local maximum of $E^2(\bm{\beta})$, there must exist a positive number \(\delta\) such that for any $\bm{\beta}$ that satisfies  \(\|\bm{\beta} - \bm{\beta}^0\| \leq \delta\), \(E(\bm{\beta}) \le E(\bm{\beta}^0)\) holds.
The neighborhood of \(\bm{\beta}^0\), i.e., \(\|\bm{\beta} - \bm{\beta}^0\| \leq \delta\), corresponds to the neighborhood of \({\eta}^0\) which is denoted as \(|\eta-\eta^0| \leq \eta'\). 
However, since \(0<{\eta^0}<{\pi}\), if we take \(\eta^1 = \eta^0 -\eta''\), where $0<\eta''<\eta'$ and $\eta''$ is small enough to ensure that \({\eta^1}>0\) is satisfied, we will have \(\cos\eta^1>\cos\eta^0\), which further implies \(E^2(\bm{\beta}^1)> E^2(\bm{\beta}^0)\), 
where $\bm{\beta}^1$ is the vector corresponding to $\eta^1$, analogous to how $\bm{\beta}^0$ corresponds to $\eta^0$. 
This contradicts the assumption that \(E^2(\bm{\beta}^0)\) is a local maximum value.

\textbf{Case}\;\textbf{\uppercase\expandafter{\romannumeral2}}:
If all $N$ complex numbers are collinear but do not satisfy condition (\ref{Econ}), that is, there must exist at least one positive integer \(1\le{t}\le{N-1}\) such that \(\beta_t^0-\beta_{t+1}^0=(2k+1){\pi}\), $k\in\mathbb{Z}$. 
%
In this case, we can verify from (\ref{Egradient}) that $\nabla E^2(\bm{\beta})=\bm{0}$ is satisfied. However, since \(\beta_t^0-\beta_{t+1}^0=(2k+1){\pi}\), we have \(\cos(\beta_t^0-\beta_{t+1}^0)=-1 \), which allows for the existence of a vector  $\hat{\mathbf{x}} \in \mathbb{R}^N$ such that forces \(-\rho_t\rho_{t+1}\cos(\beta_t^0-\beta_{t+1}^0)(\hat{x}_t-\hat{x}_{t+1})^2 \) go to positive infinity. 
Therefore,  $\hat{\mathbf{x}}^T \nabla^2 \left(E^2(\bm{\beta}^0)\right) \hat{\mathbf{x}}>0$ implies that $\nabla^2 \left(E^2(\bm{\beta}^0)\right)$ is not a negative semi-definite matrix, which also contradicts the assumption that $\bm{\beta}^0$ is a local maximum.

In summary, based on the above two cases, the uniqueness of the local maximum value $ (\sum\nolimits_{i=1}^{N}\rho_i)^2$ is proved. Thus, Lemma 4 can be proved by combining the proofs of the existing and uniqueness of the local maximum value $ (\sum\nolimits_{i=1}^{N}\rho_i)^2$ of \( E^2(\bm{\beta}) \).
\vspace{-0.3cm}
\section{Verification for LICQ}\vspace{-0.1cm}
\textit{Definition (LICQ \cite{book1}):}
Given a point \( \mathbf{x} \) and an active set of a given nonlinear optimization problem, we say that the LICQ holds if the set of active constraint gradients 
at \( \mathbf{x} \) is linearly independent.

According to the above definition, we now compute the gradients of the inequality constraints in problem (CP), i.e., $\nabla_{\mathbf{x}}^T{\left(\mathbf{U}\mathbf{x}-\mathbf{l}_u\right)}=\mathbf{U}$,
where \(\mathbf{U}\) is an \( N_t \times N_t \) matrix given in (\ref{q16}) with \(\text{Rank}(\mathbf{U}) = N_t - 1\), and it can be easily verified that any 
\( N_t-1 \) rows in \(\mathbf{U}\) are linearly independent. 
Based on the properties of \( \mathbf{U} \), we can infer that any \( c \) rows of \( \mathbf{U} \) must also be linearly independent, which indicates that problem (CP) satisfies the LICQ condition.

\vspace{-0.2cm}
{{
\section{KKT Conditions Verification for BT-BFS}\vspace{-0.1cm}
		Now we rewrite problem (CP) as follows:\vspace{-0.1cm}
\begin{equation}\label{CPrewrite}\small
	\max_{\mathbf x\in\mathbb R^{N_t}}\,g(\mathbf x)
	\quad\text{s.t.}\quad
	h_i(\mathbf x)\le0,\;i=1,\dots,N_t,\vspace{-0.2cm}
\end{equation}
with\vspace{-0cm}
\begin{equation}\small
	h_i(\mathbf x)=
	\begin{cases}
		x_i - x_{i+1} + d, &1\le i\le N_t-1,\\
		x_{N_t}-x_1 - D_x, &i=N_t.
	\end{cases}\vspace{-0.2cm}
\end{equation}
Let \(\mathcal{A}_c=\{a_1,\dots,a_c\}\subset\{1,\dots,N_t\}\) denote the indices of the \(c\) active constraints at a local maximum point \(\mathbf x^*\), and let \(\mathcal I=\{1,\dots,N_t\}\setminus\mathcal A\) denote its complement.
Define \(\mathbf{P} \in \mathbb{R}^{N_t \times N_t}\) as a permutation matrix such that 
$	\mathbf y = \mathbf P^T \mathbf x
\Longleftrightarrow
\mathbf x = \mathbf P \mathbf y,$
where \(\mathbf{P}\) rearranges the elements of \(\mathbf{x}\) so that the first \(c\) entries of \(\mathbf{y}\) correspond to the variables associated with the \(c\) active constraints, and the remaining \(N_t - c\) entries correspond to the inactive constraints.\footnote{{Note that \(\mathbf{P}\) is not unique, as any permutation that groups the active variables first and inactive variables afterward is valid; the order within each group does not affect the analysis.}}
Accordingly, we partition \(\mathbf{y}\) as \vspace{-0.2cm}
\begin{equation}\small
	\mathbf y = \begin{pmatrix}\mathbf u\\ \mathbf v\end{pmatrix},\vspace{-0.2cm}
\end{equation}
where $\mathbf u \triangleq [x_{a_1}, \dots, x_{a_c}]^T$ consists of the first \(c\) entries of \(\mathbf{y}\), while the remaining $N_t - c$ entries, $\{x_j \mid j \in \mathcal{I}\}$, form the vector $\mathbf v \in \mathbb{R}^{N_t - c}$.
We further define $\hat g (\mathbf y) \triangleq g(\mathbf P \mathbf y)$ and $\hat h_i (\mathbf y) \triangleq h_i(\mathbf P \mathbf y),1 \le i \le N_t$.
Since $\mathbf  P$ is an invertible linear operator, the mapping of variables $\mathbf x = \mathbf P \mathbf  y$ is bijective; hence, if $\mathbf  x^*$ is a local optimum of the original problem, then $\mathbf y^* = \mathbf P^{T} \mathbf  x^*$ is a local optimum of the following transformed problem:\vspace{-0.2cm}
\begin{equation}\label{yspaceProblem}\small
	\max_{\mathbf y\in\mathbb R^{N_t}}\;\hat g(\mathbf y)
	\quad\text{s.t.}\quad
	\hat h_i(\mathbf y)\le0,\quad i=1,\dots,N_t.\vspace{-0.2cm}
\end{equation}
Now we are ready to show that the proposed BT-BFS algorithm indeed satisfies the KKT conditions.

First, we verify the stationarity condition of the KKT conditions as stated in (\ref{kkt1}).
The Lagrangian of problem (\ref{CPrewrite}) is formed over the active constraints as 
$\mathcal L(\mathbf x,\boldsymbol\lambda_{c})
= g(\mathbf x) + \sum_{i\in\mathcal{A}_c}\lambda_{c,i}\,h_i(\mathbf x),$
and the stationarity \(\nabla_{\mathbf x}\mathcal L (\mathbf x^*,\boldsymbol\lambda_{c})=0\) gives\vspace{-0.1cm}
\begin{equation}\small
	\nabla g(\mathbf x^*) + \sum_{i\in\mathcal{A}_c}\lambda_{c,i}\,\nabla h_i(\mathbf x^*)=\mathbf 0.\vspace{-0.2cm}
\end{equation}
Since \(\nabla \hat g (\mathbf y)  =\mathbf P^T\nabla g(\mathbf P \mathbf y)\), the stationarity 
in \(\mathbf y\)-space is\vspace{-0.1cm}
\begin{equation}\small
	\mathbf P^T\nabla g(\mathbf P\,\mathbf y^*) \;+\;\sum_{i\in\mathcal{A}_c}\lambda_{c,i}\,\mathbf P^T\nabla h_i(\mathbf P\,\mathbf y^*)=\mathbf 0,\vspace{-0.1cm}
\end{equation}
or equivalently in problem (\ref{yspaceProblem}):\vspace{-0.1cm}
\begin{equation}\label{CPrewriteLagrangian}\small
	\nabla_{\mathbf y} \hat g(\mathbf y^*)
	+ \sum_{i\in\mathcal{A}_c}\lambda_{c,i}\,\nabla_{\mathbf y} \hat h_i(\mathbf y^*)
	= \mathbf 0.\vspace{-0.2cm}
\end{equation}
To facilitate the subsequent analysis, we introduce block‐partition and rewrite (\ref{CPrewriteLagrangian}) in the corresponding 
\((\mathbf u,\mathbf v)\) block form, i.e.,\vspace{-0.1cm}
\begin{equation}\label{blockLagrangian}\small
	\begin{pmatrix}
		\nabla_{\mathbf u} \hat g(\mathbf y^*)\\
		\nabla_{\mathbf v} \hat g(\mathbf y^*)
	\end{pmatrix}
	+ \sum_{i\in\mathcal{A}_c}\lambda_{c,i}
	\begin{pmatrix}
		\nabla_{\mathbf u} \hat h_i(\mathbf y^*)\\
		\nabla_{\mathbf v} \hat h_i(\mathbf y^*)
	\end{pmatrix}
	= \mathbf 0.\vspace{-0.2cm}
\end{equation}
Similarly, the Jacobian matrix of the active constraints can be written as\vspace{-0.1cm}
\begin{equation}\small
	\mathbf H_{\mathcal{A}_c}
	=\begin{bmatrix}
		\nabla^T_{\mathbf u} \hat h_{a_1}(\mathbf y^*) & \nabla^T_{\mathbf v} \hat h_{a_1}(\mathbf y^*)\\
		\vdots & \vdots\\
		\nabla^T_{\mathbf u} \hat h_{a_c}(\mathbf y^*) & \nabla^T_{\mathbf v} \hat h_{a_c}(\mathbf y^*)
	\end{bmatrix} 
	\triangleq \bigl[\mathbf H_{\mathcal{A}_c,\mathbf u}\;\big|\;\mathbf H_{\mathcal{A}_c,\mathbf v}\bigr],
\end{equation}
which, as verified in Appendix D, has full row rank \(c\).  Thus, by combining (\ref{blockLagrangian}), the stationarity can be obtain as \vspace{-0.1cm}
\begin{equation}\small
	\begin{cases}\label{KKTseperate}
		\nabla_{\mathbf u} \hat g(\mathbf y^*) + \mathbf H_{\mathcal{A}_c,\mathbf u}^T\,\boldsymbol\lambda_{c} =\mathbf 0,\\
		\nabla_{\mathbf v} \hat g(\mathbf y^*) + \mathbf H_{\mathcal{A}_c,\mathbf v}^T\,\boldsymbol\lambda_{c} =\mathbf 0.
	\end{cases}\vspace{-0.1cm}
\end{equation}
Since the active constraints satisfy
\(\hat h_i(\mathbf y)=0\) for \(i\in\mathcal{A}_c\), i.e., \(\mathbf H_{\mathcal{A}_c} \mathbf y = \mathbf b\), where $\mathbf b\in\mathbb R^c$ collects the $c$ entries of $\mathbf l_u$ corresponding to the active constraints,
we can form the equivalent linear equation by block-partitioning \(\mathbf H_{\mathcal{A}_c}\) and \( \mathbf y \) as
$	\mathbf H_{\mathcal{A}_c,\mathbf u}\,\mathbf u
\;+\;
\mathbf H_{\mathcal{A}_c,\mathbf v}\,\mathbf v
\;=\;
\mathbf b.$
Since \(\mathbf H_{\mathcal{A}_c}\) has full row rank $c$, \(\mathbf H_{\mathcal{A}_c,\mathbf u}\in\mathbb R^{c\times c}\) is nonsingular and we can immediately obtain $	\mathbf u 
= -\,\mathbf H_{\mathcal{A}_c,\mathbf u}^{-1}
\bigl(\mathbf H_{\mathcal{A}_c,\mathbf v}\,\mathbf v - \mathbf b\bigr),$
which exactly corresponds to the constraint-elimination step in the BT-BFS algorithm.
Define the reduced objective in \(\mathbf v\)-space as\vspace{-0.2cm}
\begin{equation}\small
	\tilde g(\mathbf v)
	\triangleq \hat g\begin{pmatrix} \mathbf u(\mathbf v) \\ \mathbf v\end{pmatrix}.\vspace{-0.1cm}
\end{equation}
Under condition (\ref{Econ}), \(\mathbf v^*\) must satisfy $\nabla_{\mathbf v} \tilde g(\mathbf v^*)=\mathbf 0$. By the chain rule, we have\vspace{-0.2cm}
\begin{equation}\label{BT-BFSgradient}\small
	\begin{aligned}
		\nabla_{\mathbf v}\,\tilde g(\mathbf v^*)
		&= \nabla_{\mathbf v}\,\hat g\bigl(\mathbf y^*\bigr)
		+ \Bigl(\tfrac{\partial\mathbf u}{\partial\mathbf v}(\mathbf v^*)\Bigr)^{T}
		\nabla_{\mathbf u}\,\hat g\bigl(\mathbf y^*\bigr)\\
		&= \nabla_{\mathbf v}\,\hat g\bigl(\mathbf y^*\bigr)
		- \Bigl(\mathbf H_{\mathcal A_c,\mathbf u}^{-1}\,\mathbf H_{\mathcal A_c,\mathbf v}\Bigr)^{T}
		\nabla_{\mathbf u}\,\hat g\bigl(\mathbf y^*\bigr)
		= \mathbf 0.
	\end{aligned}\vspace{-0.1cm}
\end{equation}

Now, we will verify ($\ref{KKTseperate}$) $\Longleftrightarrow$ ($\ref{BT-BFSgradient}$), which establishes that every solution in each iteration of the BT-BFS algorithm  satisfies the KKT stationarity.
If ($\ref{KKTseperate}$) holds, then since $\mathbf H_{\mathcal A_c,\mathbf u}$ is nonsingular, 		
substituting
$
\boldsymbol\lambda_c = -\,\mathbf H_{\mathcal A_c,\mathbf u}^{-T}\,\nabla_{\mathbf u}\hat g(\mathbf y^*)
$
into 
$
\nabla_{\mathbf v}\hat g(\mathbf y^*) + \mathbf H_{\mathcal A_c,\mathbf v}^T\,\boldsymbol\lambda_c = \mathbf 0
$
yields exactly ($\ref{BT-BFSgradient}$).
Conversely, if ($\ref{BT-BFSgradient}$) holds, then setting
$\boldsymbol\lambda_c = -\,\mathbf H_{\mathcal A_c,\mathbf u}^{-T}\,\nabla_{\mathbf u}\hat g(\mathbf y^*)$ immediately recovers ($\ref{KKTseperate}$).  Hence the two systems are algebraically equivalent.

In conclusion, the BT-BFS algorithm generates points that satisfy the KKT conditions of problem (CP), thus guaranteeing its optimality.
}
}

\vspace{-0.2cm}
\section{Derivations of \( \nabla {\bar{p}_1(\mathbf{x})} , \nabla^2 {\bar{p}_1(\mathbf{x})} \) and  $\delta_1$}\vspace{-0.1cm}
For ease of exposition, we define
$
	\mathbf{z} \triangleq \bm{\Sigma} \bm{\psi}(\mathbf{x}^i) \in \mathbb{C}^{L_t}.
$
Further let \(z_p = b_p+jq_p\) denote the \(p\)-th entry of \(\mathbf{z}\),
with  \(b_p,q_p\in \mathbb{R}\). Then, \(\bar{p}_1(\mathbf{x})\) can be written as
$\bar{p}_1(\mathbf{x}) = \operatorname{Re} \left\{ \mathbf{z}^H\bm{\psi}(\mathbf{x}) \right\}
= \sum_{i=1}^{N_t}\sum_{p=1}^{L_t}  \left( b_p\cos(\alpha^px_i)-
q_p\sin(\alpha^px_i)
\right), $
where
\(
\alpha^p \triangleq \frac{2\pi}{\lambda}(\sin\theta_t^{p}+\sin\theta).
\)
Thus, for $1\leq{i{\neq}k}\leq{N_t}$, we have $\frac{\partial \bar{p}_1(\mathbf{x})}{\partial x_{i}} = \sum_{p=1}^{L_t}  
\left( 
-b_p\alpha^p\sin(\alpha^px_i)-q_p\alpha^p\cos(\alpha^px_i)
\right),  
\frac{\partial^2 \bar{p}_1(\mathbf{x})}{\partial x_{i}^2} =   \sum_{p=1}^{L_t}   
\left( 
-b_p(\alpha^p)^2\cos(\alpha^px_i)+q_p(\alpha^p)^2\sin(\alpha^px_i)
\right)$, and $ 
\frac{\partial^2 \bar{p}_1(\mathbf{x})}{\partial x_{i} \partial x_{k}} = 0.$
We observe that \( \nabla^2 \bar{p}_1(\mathbf{x}) \) is a diagonal matrix.
Thus, since \( \max\nolimits_{1\leq{i}\leq{N_t}}\{\frac{\partial^2 \bar{p}_1(\mathbf{x})}{\partial x_{i}^2}\} \mathbf{I}_{N_t} \succeq \nabla^2 \bar{p}_1(\mathbf{x}) \), we can select 
$\delta_1 = \sum_{p=1}^{L_t}  
\sqrt{b_p^2(\alpha^p)^4+q_p^2(\alpha^p)^4}$,
which can be simply obtained by $-b_p(\alpha^p)^2\cos(\alpha^px_i)+q_p(\alpha^p)^2\sin(\alpha^px_i)\le\sqrt{b_p^2(\alpha^p)^4+q_p^2(\alpha^p)^4}$.

}

\vspace{-0.2cm}

\bibliography{ref1.bib}

\vspace{-1.2cm}

\begin{IEEEbiography}[{\includegraphics[width=1in,height=1in,keepaspectratio]{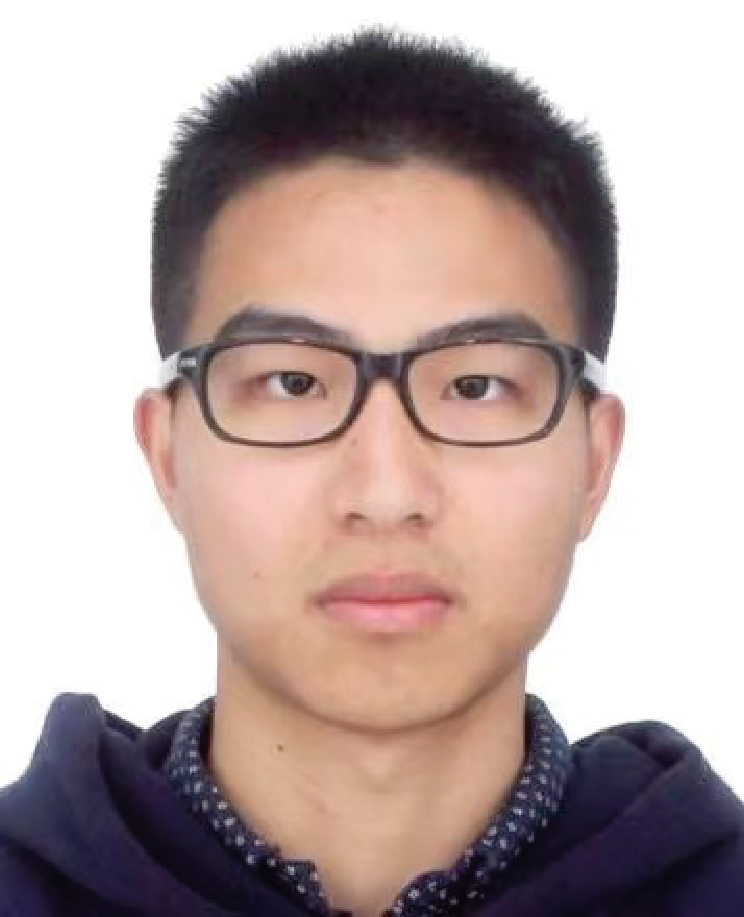}}]{Lebin Chen}
	(Student Member, IEEE) received the B.Eng. degree in information and communication engineering from Zhejiang University, Hangzhou, China, in 2024, where he is currently pursuing the Ph.D. degree with the College of Information Science and Electronic Engineering. 
	His research interests include reconfigurable MIMO, integrated sensing and communication, and signal processing for wireless communications.
\end{IEEEbiography}

\vspace{-1.2cm}

\begin{IEEEbiography}[{\includegraphics[width=1in,height=1.25in,keepaspectratio]{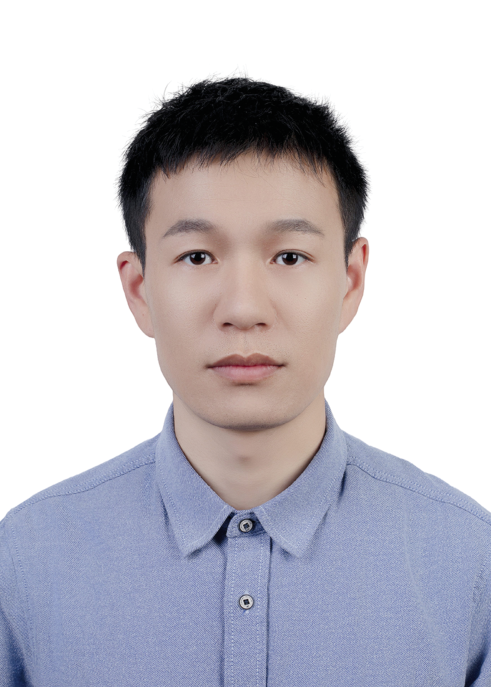}}]{Ming‑Min Zhao}
	(S’14-M’20-SM’23) received the B.Eng. and Ph.D. degrees in information and communication engineering from Zhejiang University, in 2012 and 2017, respectively. From 2015 to 2016, he was a Visiting Scholar with the Department of Electrical and Computer Engineering, Iowa State University, Ames, IA, USA. From 2017 to 2018, he worked as a Research Engineer with Huawei Technologies Co., Ltd. From 2019 to 2020, he was a Visiting Scholar with the Department of Electrical and Computer Engineering, National University of Singapore. Since 2018, he has been working with Zhejiang University, where he is currently an Associate Professor with the College of Information Science and Electronic Engineering. His research interests include algorithm design and analysis for advanced MIMO, signal processing for communication, channel coding, and machine learning for wireless communications. He was the recipient of the IEEE Communications Society Katherine Johnson Young Author Best Paper Award in 2024.
\end{IEEEbiography}

\vspace{-1.2cm}

\begin{IEEEbiography}[{\includegraphics[width=1in,height=1.25in,keepaspectratio]{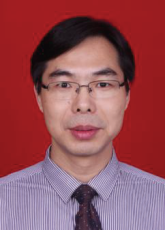}}]{Min‐Jian Zhao}
	(Senior Member, IEEE) received the M.Sc.\ and Ph.D.\ degrees in communication and information systems from Zhejiang University, Hangzhou, China, in 2000 and 2003, respectively.
	
	He was a Visiting Scholar with the University of York, York, U.K., in 2010. He is currently a Professor with the College of Information Science and Electronic Engineering, Zhejiang University. His current research interests include modulation theory, channel estimation and equalization, MIMO signal processing for wireless communications, anti‐jamming technology for wireless transmission and networking, and communication SoC chip design.
\end{IEEEbiography}

\vspace{-1.2cm}

\begin{IEEEbiography}[{\includegraphics[width=1in,height=1.25in,keepaspectratio]{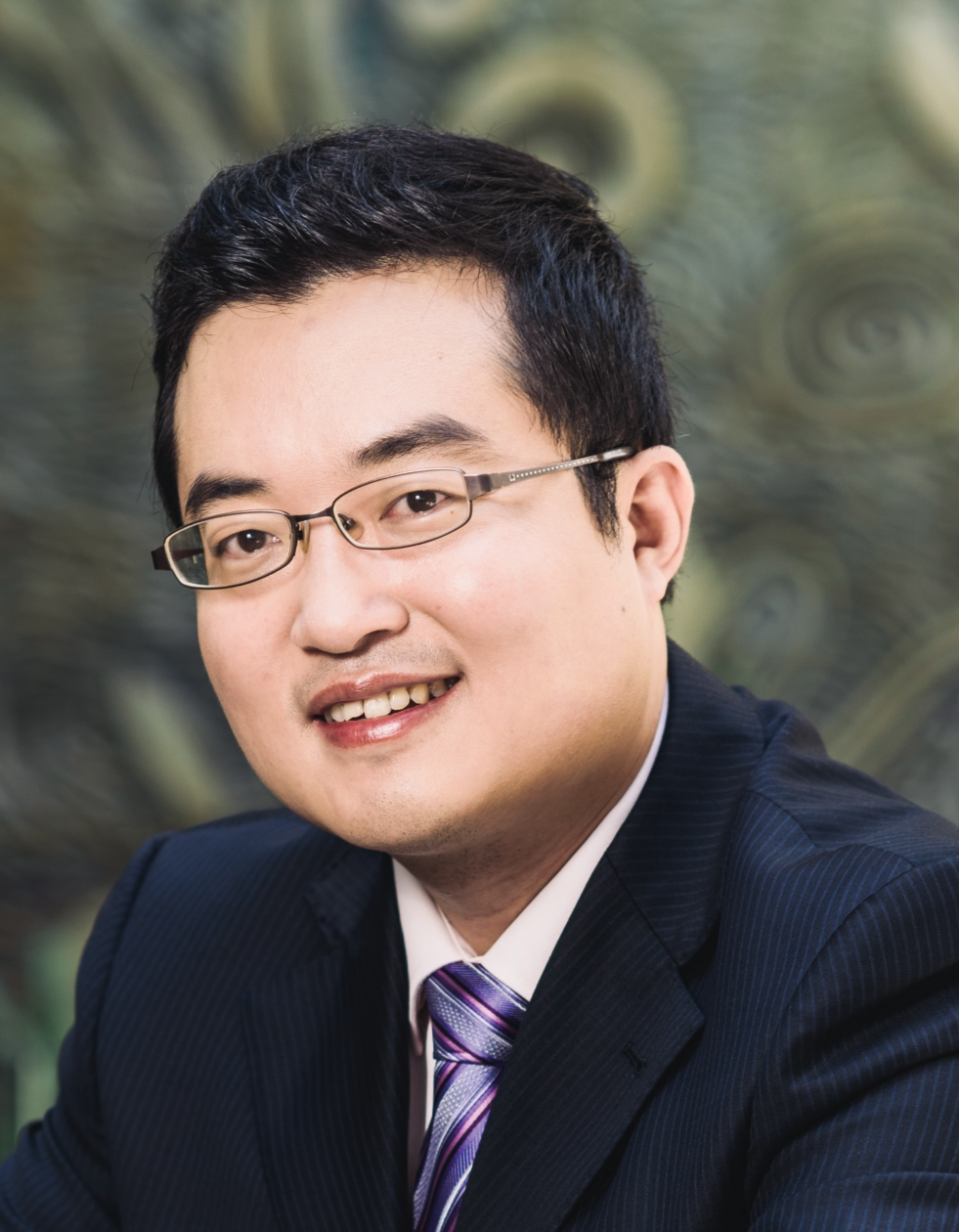}}]{Rui Zhang} (S'00-M'07-SM'15-F'17) received the B.Eng. (first-class Hons.) and M.Eng. degrees from the National University of Singapore, Singapore, and the Ph.D. degree from the Stanford University, Stanford, CA, USA, all in electrical engineering.
From 2007 to 2009, he worked as a researcher at the Institute for Infocomm Research, ASTAR, Singapore. In 2010, he joined the Department of Electrical and Computer Engineering of National University of Singapore, where he was appointed as a Provost’s Chair Professor in 2020. He is now with the School of Science and Engineering, Shenzhen Research Institute of Big Data, The Chinese University of Hong Kong, Shenzhen, as a Principal’s Diligence Chair Professor. He has published over 350 journal papers and over 200 conference papers. He has been listed as a Highly Cited Researcher by Thomson Reuters/Clarivate Analytics since 2015. His current research interests include UAV/satellite communications, wireless power transfer, intelligent reflecting surface, reconfigurable MIMO, radio mapping and optimization methods.      

He was the recipient of the 6th IEEE Communications Society Asia-Pacific Region Best Young Researcher Award in 2011, the Young Researcher Award of National University of Singapore in 2015, the Wireless Communications Technical Committee Recognition Award in 2020, the IEEE Signal Processing and Computing for Communications (SPCC) Technical Recognition Award in 2021, and the IEEE Communications Society Technical Committee on Cognitive Networks (TCCN) Recognition Award in 2023. His works received 18 IEEE Best Journal Paper Awards, including the IEEE Marconi Prize Paper Award in Wireless Communications in 2015 and 2020, the IEEE Signal Processing Society Best Paper Award in 2016, the IEEE Communications Society Heinrich Hertz Prize Paper Award in 2017, 2020 and 2022, the IEEE Communications Society Stephen O. Rice Prize in 2021, etc. He served for over 30 international conferences as the TPC co-chair or an organizing committee member. He was an elected member of the IEEE Signal Processing Society SPCOM Technical Committee from 2012 to 2017 and SAM Technical Committee from 2013 to 2015, and served as the Vice Chair of the IEEE Communications Society Asia-Pacific Board Technical Affairs Committee from 2014 to 2015. He was a Distinguished Lecturer of IEEE Signal Processing Society and IEEE Communications Society from 2019 to 2020. He served as an Editor for the IEEE TRANSACTIONS ON WIRELESS COMMUNICATIONS from 2012 to 2016, the IEEE JOURNAL ON SELECTED AREAS IN COMMUNICATIONS: Green Communications and Networking Series from 2015 to 2016, the IEEE TRANSACTIONS ON SIGNAL PROCESSING from 2013 to 2017, the IEEE TRANSACTIONS ON GREEN COMMUNICATIONS AND NETWORKING from 2016 to 2020, and the IEEE TRANSACTIONS ON COMMUNICATIONS from 2017 to 2022. He served as a member of the Steering Committee of the IEEE Wireless Communications Letters from 2018 to 2021. He is a Fellow of the Academy of Engineering Singapore. 
	\end{IEEEbiography}

\bibliographystyle{IEEEtran}

\end{document}